\documentclass{siamart1116}
\usepackage{amsfonts}
\usepackage{amssymb}
\usepackage{graphicx}
\usepackage{algorithmic}
\ifpdf
  \DeclareGraphicsExtensions{.eps,.pdf,.png,.jpg}
\else
  \DeclareGraphicsExtensions{.eps}
\fi
\numberwithin{theorem}{section}
\newtheorem{assumption}{Assumption}
\numberwithin{assumption}{section}
\newtheorem{property}{Property}
\newtheorem{remark}{Remark}
\newcommand{\TheTitle}{Isogeometric Mortar Coupling\\for Electromagnetic Problems}
\newcommand{\TheTitleShort}{IGA Mortar Coupling for EM Problems} 
\newcommand{\TheAuthors}{A. Buffa, J. Corno, C. de Falco, S. Sch\"ops, R. V\'azquez}
\headers{\TheTitleShort}{\TheAuthors}
\title{{\TheTitle}\thanks{Submitted.
\funding{J. Corno's work is supported by the DFG Grant SCHO1562/3-1 and the "Excellence Initiative" of the German Federal and State Governments and the GSCE at TU Darmstadt.}}}
\author{
  Annalisa Buffa\thanks{ Institute of Mathematics, Ecole Polytechnique F\'ed\'erale de Lausanne, Lausanne, Switzerland (\email{annalisa.buffa@epfl.ch}).}
   \and
   Jacopo Corno\thanks{Graduate School CE and Institut f\"ur Theorie und Electromagnetischer Felder, TU Darmstadt, Darmstadt, Germany
     (\email{[corno,schoeps]@gsc.tu-darmstadt.de}).}
   \and
   Carlo de Falco\thanks{MOX Modeling and Scientific Computing, Politecnico di Milano, Milano, Italy
     (\email{carlo.defalco@polimi.it}).}
   \and
   Sebastian Sch\"ops\footnotemark[3]
   \and
   Rafael V\'azquez\footnotemark[2]\ \thanks{Istituto di Matematica Applicata e Tecnologie Informatiche del CNR, Via Ferrata, 5, 27100 Pavia, Italy (\email{vazquez@imati.cnr.it})}
 }

\usepackage{amsopn}

\usepackage{siunitx}
\usepackage{nicefrac}
\usepackage{xcolor,colortbl}
\usepackage{booktabs}

\usepackage{bm}
\usepackage{tikz}
\usepackage{pgfplots}
\pgfplotsset{compat=1.13}
\usepgfplotslibrary{colorbrewer}
\usepackage{color}
\definecolor{tud1d}{RGB}{36,53,114}
\definecolor{tud3d}{RGB}{0,113,94}
\definecolor{tud6d}{RGB}{174,142,0}
\definecolor{tud8d}{RGB}{169,73,19}
\definecolor{tud10d}{RGB}{115,32,84}
\usepackage{subcaption}
\usepackage{tikz-cd}
\usepackage{etoolbox}
\usepackage[acronym]{glossaries}
\newacronym{rf}{RF}{Radio Frequency}
\newacronym{pde}{PDE}{Partial Differential Equation}
\newacronym{fe}{FE}{Finite Element}
\newacronym{fem}{FEM}{Finite Element Method}
\newacronym{bem}{BEM}{Boundary Element Method}
\newacronym{iga}{IGA}{IsoGeometric Analysis}
\newacronym{cad}{CAD}{Computer Aided Design}
\newacronym{brep}{B-Rep}{Boundary Representation}
\newacronym{nurbs}{NURBS}{Non-Uniform Rational B-Splines}
\newacronym{mor}{MOR}{Model Order Reduction}
\newacronym{rom}{ROM}{Reduced Order Model}
\newacronym{linac}{linac}{LINear ACcelerator}
\newacronym{tesla}{TESLA}{TeV-Energy Superconducting Linear Accelerator}
\newacronym{ttf}{TTF}{TESLA Test Facility}
\newacronym{ilc}{ILC}{International Linear Collider}
\newacronym{desy}{DESY}{Deutsches Elektronen-Synchrotron}
\newacronym{hom}{HOM}{Higher Order Mode}
\newacronym{homc}{HOMC}{Higher Order Mode Coupler}
\newacronym{te}{TE}{Transverse Electric}
\newacronym{tm}{TM}{Transverse Magnetic}
\newacronym{pec}{PEC}{Perfect Electric Conducting}
\newacronym{pmc}{PMC}{Perfect Magnetic Conducting}
\newacronym{ssc}{SSC}{State Space Concatenation}
\newacronym{cst}{CST\textsuperscript{\textregistered}}{Computer Simulation Technology}
\newacronym{uq}{UQ}{Uncertainty Quantification}
\newacronym{dd}{DD}{Domain Decomposition}
\newacronym{ddm}{DDM}{Domain Decomposition Method} 
\newcommand{\epsvac}{\ensuremath{\varepsilon_0}}
\newcommand{\muvac}{\ensuremath{\mu_0}}
\newcommand{\vect}[1]{\ensuremath{\bm{\mathbf{#1}}}}

\newcommand{\R}{\ensuremath{\mathbb{R}}}
\newcommand{\Ltwo}[1]{\ensuremath{L^{2}(#1)}}
\newcommand{\Ltwovec}[1]{\ensuremath{\mathbf{L}^{2}(#1)}}

\newcommand{\Hdiv}[2]{\ensuremath{\mathbf{H}_{#2}(\mathrm{div};#1)}}
\newcommand{\Hcurl}[2]{\ensuremath{\mathbf{H}_{#2}(\mathbf{curl};#1)}}
\newcommand{\transpose}{^{\top}}
\newcommand{\grad}{\ensuremath{\mathbf{grad} \,}}
\newcommand{\curl}[1]{\ensuremath{\, \mathbf{curl} \, #1}}
\newcommand{\curls}[1]{\ensuremath{\, \mathrm{curl} \, #1}}
\renewcommand{\div}[1]{\ensuremath{\, \mathrm{div} \, #1}}
\newcommand{\gradG}{\ensuremath{\mathbf{grad}_\Gamma \,}}
\newcommand{\curlG}[1]{\ensuremath{\, \mathbf{curl}_\Gamma \, #1}}
\newcommand{\curlsG}[1]{\ensuremath{\, \mathrm{curl}_\Gamma \, #1}}
\newcommand{\divG}[1]{\ensuremath{\, \mathrm{div}_\Gamma \, #1}}
\newcommand{\laplace}{\ensuremath{\Delta}}

\newcommand{\scalar}[2]{\ensuremath{\left(#1,#2\right)}}
\newcommand{\traceh}{\gamma}
\newcommand{\trace}{\vect{\gamma}}
\newcommand{\truxn}{\vect{\gamma}_\perp}
\newcommand{\traceg}[1]{\vect{\gamma}_{#1}}
\newcommand{\truxng}[1]{\vect{\gamma}_{\perp,#1}}
\newcommand{\onref}[1]{\ensuremath{\widehat{#1}}}
\newcommand{\map}{\ensuremath{\mathbf{F}}}
\newcommand{\mapjac}{\ensuremath{D\map}}
\newcommand{\Sonestar}[2]{S_{#1}^{1^*}(#2)}
\newcommand{\Spone}[1]{S_p^1(#1)}
\newcommand{\Sponestar}[1]{\Sonestar{p}{#1}}
\newcommand{\Vfem}[1]{X_h(#1)}
\newcommand{\Ndof}{\ensuremath{N_{\text{DoF}}}}
\newcommand{\Ngamma}{\ensuremath{N_{\Gamma}}}
\newcommand{\betainfsup}{\ensuremath{\beta_{\text{infsup}}}}

 \ifpdf
\hypersetup{
  pdftitle={\TheTitle},
  pdfauthor={\TheAuthors}
}
\fi

\begin{document}
\maketitle

\begin{abstract}
  This paper discusses and analyses two domain decomposition approaches for electromagnetic problems that allow the combination of domains discretised by either N\'ed\'elec-type polynomial finite elements or spline-based isogeometric analysis. The first approach is a new isogeometric mortar method and the second one is based on a modal basis for the Lagrange multiplier space, called state-space concatenation in the engineering literature. Spectral correctness and in particular inf-sup stability of both approaches are analytically and numerically investigated. The new mortar method is shown to be unconditionally stable. Its construction of the discrete Lagrange multiplier space takes advantage of the high continuity of splines, and does not have an analogue for N\'ed\'elec finite elements. On the other hand, the approach with modal basis is easier to implement but relies on application knowledge to ensure stability and correctness.
\end{abstract}



\section{Introduction}
\label{sec:intro}

In electrical engineering numerical modelling and simulations have become more and more invaluable in the design process of new devices and components. In particular, we are often interested in electromagnetic devices where the geometry plays an important role in their performance. Given this requirement, we choose to investigate the applicability of \gls{iga} to the simulation of complex electromagnetic structures such as, e.g., \gls{rf} cavities (see Fig.~\ref{fig:tesla_cavity}) as used in particle accelerators. The main building block of IGA for electromagnetic problems are the B-spline spaces with \gls{nurbs} mappings as introduced by Buffa et al. in~\cite{Buffa_2010aa,Buffa_2011aa}. They allow us to parametrise our domain of interest exactly in terms of \gls{cad}, thus avoiding geometrical errors, and grants us a straightforward way to deal with deformations or shape optimisation processes, e.g. \cite{Bontinck_2017ag}. 

However, it is known that, when dealing with complicated structures, one might not always be able to easily construct a volumetric parametrisation as required by \gls{iga}, and no robust automatic tools exist to alleviate the burden of this task. It is then of interest to investigate the possibility of substructuring and coupling domains, possibly using different type of discretisation with each other. A particular example of electromagnetic devices with problematic parametrisation are electrical machines where the rotor part of the domain is rotating with respect to the stator~\cite{Bontinck_2018ac}.

To address those issues, we investigate \gls{ddm} for electromagnetic problems solved with isogeometric methods. In particular, we are interested in methods that allow for the coupling of non-conforming meshes and of different discretisation schemes such as \gls{iga} and the classical N\'ed\'elec-type \gls{fem}. Two methods are considered: a new mortar method, which allows for the coupling of different grids and exploits the inherent properties of the isogeometric basis to naturally define the approximation space for the Lagrange multipliers, and the \gls{ssc} method recently introduced by Flisgen et al.~\cite{Flisgen_2013aa}, which exploits a modal basis on the connecting interfaces instead. \gls{ssc} can also be interpreted as a problem-specific port reduction method \cite{Eftang_2014aa,Iapichino_2012aa}. In the following, we discuss the proper mathematical construction of such basis, its stability, correctness and show numerical simulations of a real world application example.

\subsection{Application to Radio Frequency Cavities}

The motivational application example of this paper is the simulation of \gls{rf} cavities which are used to give energy to the beam in particle accelerators. A resonating electromagnetic field is induced inside these structures in such a way that the field oscillation is synchronous with the passing of the charges and that they experience only an accelerating field~\cite{Wangler_2008aa}.

More specifically, we consider the \gls{tesla} cavity~\cite{Aune_2000aa, TTF_1995}, a 9-cell cavity built in superconducting niobium for the \gls{ilc}, currently in operation at the \gls{desy} facility in Hamburg (see Fig.~\ref{fig:tesla_cavity}). The \gls{tesla} cavity operates with a \gls{tm} standing wave mode at $f_0 = \SI{1.3}{\giga\hertz}$. Very stringent tolerancies are required in production in order for the accelerating frequency to be as close as possible to this value. This high precision is also required from the numerical simulation.

\begin{figure}
\centerline{\includegraphics[width=.8\textwidth]{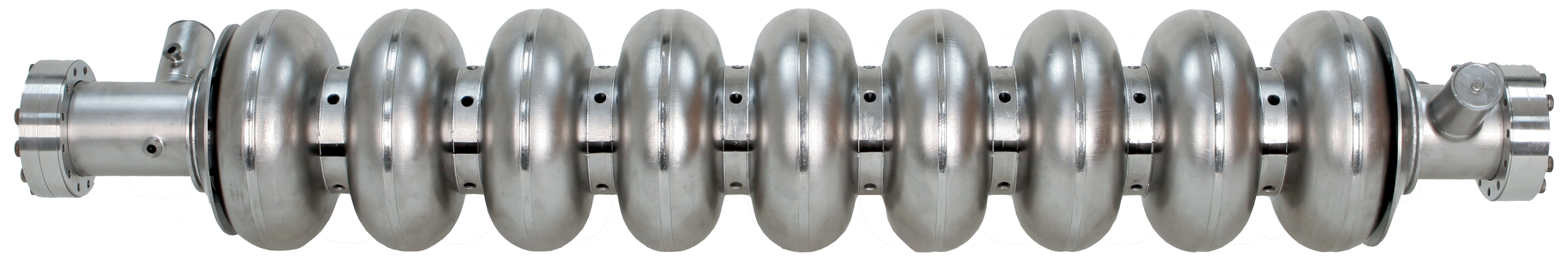}}
\caption{A superconducting \gls{tesla} cavity. (Copyright 2006 \gls{desy})\label{fig:tesla_cavity}}
\end{figure}

However, many cavity simulation codes still rely on 2D axysimmetric \gls{fem}, thus disregarding the 3D effects of the \gls{homc} present at the ends of the cavity (see Fig.~\ref{fig:tesla_cavity}-\ref{fig:homc}). Three dimensional \gls{fem} needs a very high number of elements in order to achieve a sufficient accuracy. The use of \gls{iga} in cavity simulation has been proven to be beneficial both in terms of accuracy and of overall reduction of the computational cost~\cite{Corno_2016aa, CornoThesis}. Moreover, \gls{iga} allows for a better treatment of geometry deformations, e.g. due to Lorentz detuning \cite{Corno_2016aa}. 

The overall goal is to be able to discretise the central part of the cavity with an isogeometric scheme (as this is the region where the geometry description plays a paramount role in the definition of the eigenfrequencies and where the wall deformations occur) and \gls{fe} in the end beampipes, where meshing is most problematic and (possibly curved) tetrahedra allow for the inclusion of the fine details of the \glspl{homc} (see Fig.~\ref{fig:homc}). Furthermore, given the modularity of the geometry, we consider the application of \gls{ddm} also for cell-to-cell coupling.

\medskip

The paper is organised as follows. Maxwell's eigenvalue problem in the cavity is recalled in Section~\ref{sec:maxwellIGA}, along with the \gls{iga} framework used. Two \gls{dd} methods, the mortar method and the \gls{ssc} method, are presented in Section~\ref{sec:substructuring}, followed by a complete analysis of the mortar method in Section~\ref{sec:analysis}, and numerical tests in Section~\ref{sec:results}. Section~\ref{sec:conclusions} draws some conclusions and final remarks.

\begin{figure}
\centerline{\includegraphics[width=.8\textwidth]{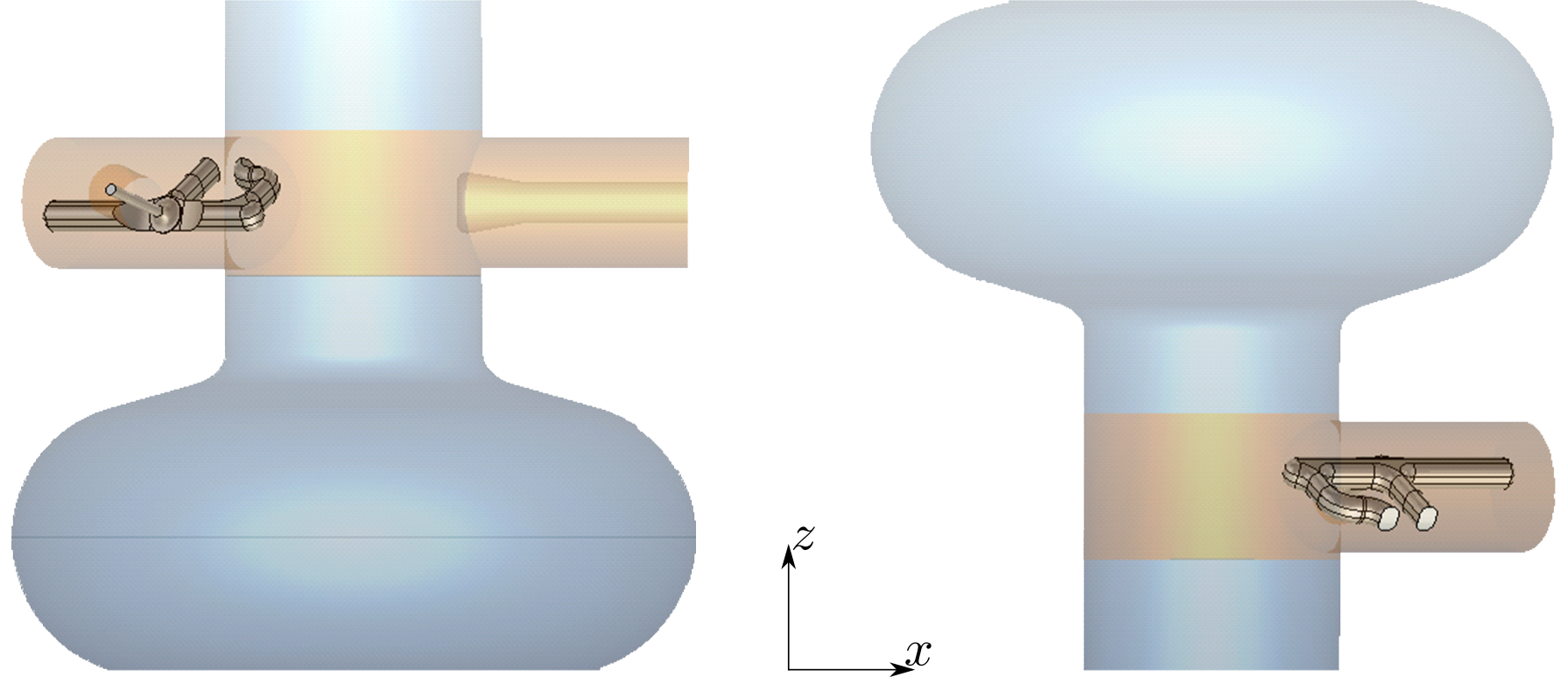}}
\caption{\glspl{homc} of the \gls{tesla} cavity. On the left the downstream coupler, on the right the upstream one.\label{fig:homc} Images are created with CST Microwave Studio \cite{CST_2018aa}}
\end{figure}
 
\section{Isogeometric Analysis for Maxwell's Equations}
\label{sec:maxwellIGA}
This section introduces first notation, Sobolev spaces and trace operators. The Maxwell eigenvalue problem is formulated in its weak form in the second subsection. Although focusing on the eigenvalue problem, all methods can be applied to source problems as well.

\subsection{Sobolev spaces and trace operators} \label{sec:sobolev}
We distinguish between two-di\-men\-sional and three-dimensional domains by using the notation $\Gamma$ and $\Omega$, respectively. We denote by $H^s(\Gamma)$ and $H^s(\Omega)$ the Sobolev spaces of regularity $s \in \mathbb{R}$, and by $\| \cdot \|_{s,\Gamma}$ and $\| \cdot \|_{s,\Omega}$ the corresponding norms. The domain subindices will be removed from the norms when there is no ambiguity. Moreover, assuming that $\Gamma \subset \partial \Omega$, we follow the notation introduced in \cite{Lions-Magenes} and make use of the space $H^{1/2}_{00}(\Gamma)$ of functions such that their extension by zero is in $H^{1/2}(\partial \Omega)$, and denote its dual by $H^{-1/2}_{00}(\Gamma)$. For vector fields, the spaces will be denoted with bold letters, for example, ${\bf H}^s(\Omega)$.

We will also make use of the spaces of vector fields
\begin{equation*}
\begin{aligned}
\Hcurl{\Omega}{} &= \{ \vect{v} \in \Ltwovec{\Omega} : \curl \vect{v} \in \Ltwovec{\Omega} \}, \\
\Hdiv{\Omega}{}  &= \{ \vect{v} \in \Ltwovec{\Omega} : \div \vect{v} \in \Ltwo{\Omega} \}.
\end{aligned}
\end{equation*}
Denoting by $\vect{n}$ the unit normal vector exterior to $\Omega$, for scalar fields we introduce the standard trace operator $\traceh: H^1(\Omega) \longrightarrow H^{1/2}(\partial \Omega)$, while for vector fields in $\Hcurl{\Omega}{}$ we introduce the two trace operators
\begin{equation*}
\begin{aligned}
\trace(\vect{u}) &= \vect{u} - (\vect{u} \cdot\vect{n}) \vect{n} \equiv (\vect{n} \times \vect{u}) \times \vect{n}, \\
\truxn (\vect{u}) &= \vect{u} \times \vect{n},
\end{aligned}
\end{equation*}
and define $\Hcurl{\Omega}{0}$ as the subspace of functions in $\Hcurl{\Omega}{}$ with vanishing trace $\trace$ on the boundary $\partial \Omega$. We recall that the trace operators and differential operators commute, and in particular $\gradG (\traceh \phi) = \trace (\grad \phi)$, see for instance \cite{Buffa_2018}. Finally, for $\Gamma \subset \partial \Omega$ we denote by $\traceh_\Gamma$, $\traceg{\Gamma}$ and $\truxng{\Gamma}$ the restriction of the trace operators to $\Gamma$, and denote by $\Hcurl{\Omega}{\Gamma}$ the functions of $\Hcurl{\Omega}{}$ with vanishing trace on $\Gamma$. In particular, the trace operators $\traceg{\Gamma}$ and $\truxng{\Gamma}$ map the functions in $\Hcurl{\Omega}{}$ into the Sobolev spaces
\begin{align*}
\mathbf{H}^{-1/2}(\mathbf{curl};\Gamma) &= \{ \vect{v} \in \mathbf{H}^{-1/2}(\Gamma) : \curlsG{\vect{v}} \in H^{-1/2}(\Gamma) \}, \\
\mathbf{H}^{-1/2}(\mathbf{div};\Gamma)  &= \{ \vect{v} \in \mathbf{H}^{-1/2}_{00}(\Gamma) : \divG{\vect{v}} \in H^{-1/2}_{00}(\Gamma) \},
\end{align*}
respectively. We will endow these spaces with the usual graph norms, that we respectively denote by $\| \cdot \|_{-1/2,\curls{}}$ and, for simplicity, $\| \cdot \|_{-1/2,\div{}}$.

\subsection{Formulation of Maxwell eigenvalue problem}
Under the assumption of time-harmonicity, the electric field in a \gls{rf} cavity is governed by a second order \gls{pde} in terms of the electric field phasor $\vect{E}$ only
\begin{equation}\label{eq:curl-curl}
\curl{\curl{\vect{E}}} = \epsvac\muvac\omega^2\vect{E},
\end{equation}
where $\epsvac$ and $\muvac$ are the electric permittivity and magnetic permeability of vacuum respectively \cite{Jackson_1998aa}. Equation~\eqref{eq:curl-curl} is an eigenvalue problem, whose solutions are a sequence of eigenmodes $(\omega_n^2,\vect{E}_n)$ which represent the excitable modes in the cavity. In the lossy case, the modal analysis leads to solutions in the complex plane. In this work we focus on the lossless approximation where the cavity walls are considered to be perfect conductors, which is a reasonable assumption for superconducting resonators such as the \gls{tesla} cavity. In this case all solutions are real.

Following~\cite{Boffi_2010aa}, we introduce the weak formulation of Maxwell's eigenproblem~\eqref{eq:curl-curl} as:
\emph{Find $\omega\in\R$ and }$\vect{E}\in\Hcurl{\Omega}{0}$\emph{, with $\vect{E}\neq 0$, such that}
\begin{align}\label{eq:curl-curl-weak}
\scalar{\curl{\vect{E}}}{\curl{\vect{v}}} = \epsvac\muvac\omega^2\scalar{\vect{E}}{\vect{v}} && \forall \vect{v}\in\Hcurl{\Omega}{0}.
\end{align}
In the context of classical \gls{fe} analysis, the numerical approximation of problem~\eqref{eq:curl-curl-weak} requires either some form of stabilisation of the divergence part or the use of the so-called edge elements introduced by N\'ed\'elec~\cite{Nedelec_1980aa}, which have the property of directly satisfying the commuting de Rham diagram~\cite{Monk_2003aa}.

In~\cite{Buffa_2010aa,Buffa_2011aa}, Buffa et al. introduce a sequence of B-spline spaces that also satisfies the de Rham diagram, opening up the possibiliy to apply \gls{iga} to electromagnetic problems. In the remainder of this section, we introduce some notation regarding B-spline spaces and introduce the discretisation scheme used.

\subsection{B-spline Basis Functions}

Given a degree $p$, B-splines are defined from a so-called knot vector
\begin{equation*}
\Xi = \left[\xi_1, \dots, \xi_{n+p+1}\right] \quad 0 \leq \xi_i \leq \xi_{i+1} \leq 1,
\end{equation*}
using the Cox-De Boor recursion formula~\cite{Piegl_1997aa}, where $n$ is the number of functions, see Fig.~\ref{fig:B-spline_basis}. We assume that the knot vector is open, which means that the first and last knots are repeated $p+1$ times, and that all the internal knots are repeated at most $p$ times. We denote by $\onref{B}_i^p$ the $i$-th basis function of degree $p$ on the reference domain $(0,1)$ and define the space of B-spline as
\begin{equation*}
S_p(\Xi) = {\rm span} \left\lbrace \onref{B}_i^p, i=1,\dots,n\right\rbrace.
\end{equation*}
We will denote by $h$ the maximum size of the non-empty elements $(\xi_i,\xi_{i+1})$, and we assume that the mesh is locally quasi-uniform, that is, there exists $\theta\ge 1$ such that the size ratio for two adjacent elements satisfies $\theta^{-1} \le h_i / h_{i+1} \le \theta$.

Multivariate spaces in the reference domain are defined following a tensor product approach, e.g, in 3D, we define the degrees $p_j$, the knot vectors $\Xi_j$ and the integers $n_j$, for $j=1,2,3$, to get the basis functions
\begin{equation*}
\onref{B}_{\vect{i}}^{\vect{p}}(\vect{\xi}) = \onref{B}_{i_1}^{p_1}(\xi_1) \onref{B}_{i_2}^{p_2}(\xi_2) \onref{B}_{i_3}^{p_3}(\xi_3),
\end{equation*}
on $(0,1)^3$. Here we have defined the degree vector $\vect{p}=[p_1, p_2, p_3]$ and the multi-index $\vect{i}\in\mathcal{I}=\lbrace[i_1,i_2,i_3]:1 \leq i_j \leq n_j\rbrace$. The space spanned by multivariate B-spline is denoted by $S_{p_1,p_2,p_3}(\Xi_1,\Xi_2,\Xi_3)$ and its dimension by $N=\prod_j n_j$.

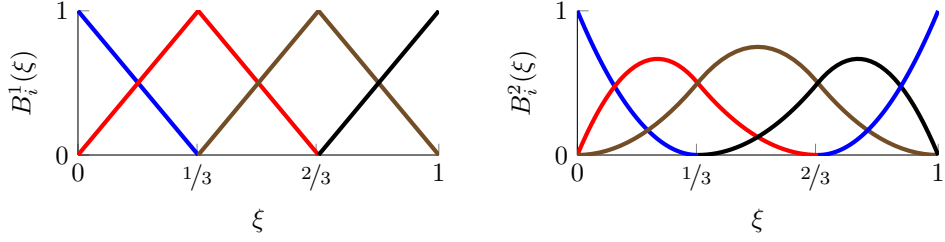
\begin{figure}
\begin{minipage}{.49\textwidth}
\begin{tikzpicture}
	\begin{axis}[axis x line*=bottom,axis y line*=left,width=\linewidth,height=3.5cm,xlabel={$\xi$},ylabel={$B_i^1(\xi)$},xmin=0,xmax=1,ymin=0,ymax=1,xtick={0,0.33,0.66,1},xticklabels={0,$\nicefrac{1}{3}$, $\nicefrac{2}{3}$,1},ytick={0,1}]
		\addplot +[ultra thick] table [mark=none, x=u1, y=B1, col sep=comma] {fig03a.csv};
		\addplot +[ultra thick] table [mark=none, x=u2, y=B2, col sep=comma] {fig03a.csv};
		\addplot +[ultra thick] table [mark=none, x=u3, y=B3, col sep=comma] {fig03a.csv};
		\addplot +[ultra thick] table [mark=none, x=u4, y=B4, col sep=comma] {fig03a.csv};
	\end{axis}
\end{tikzpicture}
\end{minipage}\hfill
\begin{minipage}{.49\textwidth}
\begin{tikzpicture}
	\begin{axis}[axis x line*=bottom,axis y line*=left,width=\linewidth,height=3.5cm,xlabel={$\xi$},ylabel={$B_i^2(\xi)$},xmin=0,xmax=1,ymin=0,ymax=1,xtick={0,0.33,0.66,1},xticklabels={0,$\nicefrac{1}{3}$, $\nicefrac{2}{3}$,1},ytick={0,1}]
		\addplot +[ultra thick] table [mark=none, x=u1, y=B1, col sep=comma] {fig03b.csv};
		\addplot +[ultra thick] table [mark=none, x=u2, y=B2, col sep=comma] {fig03b.csv};
		\addplot +[ultra thick] table [mark=none, x=u3, y=B3, col sep=comma] {fig03b.csv};
		\addplot +[ultra thick] table [mark=none, x=u4, y=B4, col sep=comma] {fig03b.csv};
		\addplot +[ultra thick] table [mark=none, x=u5, y=B5, col sep=comma] {fig03b.csv};
	\end{axis}
\end{tikzpicture}
\end{minipage}
\caption{Example of B-spline basis functions of degree 1 (on the left) and degree 2 (on the right).\label{fig:B-spline_basis}}
\end{figure}

\gls{nurbs} basis functions are defined as rational B-splines as
\begin{equation*}
\onref{N}_{\vect{i}}^{\vect{p}} = \frac{w_{\vect{i}}\onref{B}_{\vect{i}}^{\vect{p}}}{\sum_{j=1}^N w_{\vect{j}}\onref{B}_{\vect{j}}^{\vect{p}}},
\end{equation*}
with $w_{\vect{i}}$ weights associated to each basis function. In the case $w_{\vect{i}}=1\; \forall {\vect{i}}$ we revert to the B-spline case.
\gls{nurbs} geometries are built as a map from the reference domain $\onref{D} = (0,1)^d$ to the physical space by defining a control polyhedron in the physical space. Each control point $\vect{P}_{\vect{i}} \in \R^k$, with $k \ge d$ in the polyhedron is associated to a basis function, defining the map
\begin{equation*}\mathbf{F}(\vect{\xi}) = \sum_{i=1}^N \onref{N}_{\vect{i}}^{\vect{p}} \vect{P}_{\vect{i}}
\end{equation*}
from $\onref{D}$ to $D \subset \mathbb{R}^k$. As already mentioned, the domain will be denoted by $\Gamma$ when $d=2$, and by $\Omega$ when $d=3$.

\subsection{Isogeometric Discretisation}\label{sec:IGA}
We assume that our geometry is defined through a NURBS mapping $\map : \onref{\Omega} =(0,1)^3 \rightarrow \Omega \subset \mathbb{R}^3$, and to prevent singularities in the mapping we assume that the mapping is a bi-Lipschitz homeomorphism. Moreover, let us assume for simplicity that the degree is the same in all directions, i.e., $p = p_1 = p_2 = p_3$. Following~\cite{Buffa_2011aa}, we first define the spline spaces in the reference domain as
\begin{align}
& S^0_p(\onref{\Omega})= S_{p,p,p}(\Xi_1,\Xi_2,\Xi_3)\\
& \Spone{\onref{\Omega}}= S_{p-1,p,p}(\Xi_1',\Xi_2,\Xi_3)\times S_{p,p-1,p}(\Xi_1,\Xi_2',\Xi_3)\times S_{p,p,p-1}(\Xi_1,\Xi_2,\Xi_3'), \label{eq:S1hat}
\end{align}
where $\Xi_j'=[\xi_2^j,\dots,\xi_{n_j+p}^j]$ for $j=1,2,3$ is a modified knot vector with the first and last knot removed. This corresponds to lowering by one the degree and the regularity of the basis functions in the $j$-th direction. Then, considering the mapping $\map$ and its Jacobian $\mapjac$, the spline spaces in the physical domain $\Omega$ are defined by push-forward as
\begin{align}
& S_p^0(\Omega) = \{ u \in H^1(\Omega) : u \circ \map \in S^0_p(\onref{\Omega}) \}, \\
& \Spone{\Omega} = \{ \vect{u} \in \Hcurl{\Omega}{} : \left(\mapjac\right)\transpose\left(\vect{u} \circ \map\right) \in \Spone{\onref{\Omega}} \}. \label{eq:S1}
\end{align}
To deal with boundary conditions, let $\Sigma \subseteq \partial \Omega$, and for simplicity we assume that $\Sigma = {\bf F}(\onref{\Sigma})$, where $\onref{\Sigma} \subseteq \partial \onref{\Omega}$ is the union of some boundary sides of the reference domain. We introduce the discrete spaces with vanishing boundary conditions on $\Sigma \subseteq \partial \Omega$ as
\begin{align*}
& S^0_p(\Omega;\Sigma) = \{ u \in S^0_p(\Omega) : \traceh(u) = 0 \text{ on } \Sigma \}, \\
& \Spone{\Omega;\Sigma} = \{ \vect{u} \in S^1_p(\Omega) : \trace(\vect{u}) = \vect{0} \text{ on } \Sigma \},
\end{align*}
Analogously to Lagrangian and N\'ed\'elec's edge elements, the spline spaces defined above are part of a more general family of isogeometric spaces that form a discrete exact sequence. There also exists a set of commutative projectors such that they conform a commutative de Rham diagram, and it can be proved that the discretisation of the eigenvalue problem~\eqref{eq:curl-curl-weak} with the space $\Spone{\Omega;\partial \Omega}$ is spurious free. For more details we refer to \cite{Buffa_2011aa} and \cite[Ch.~5]{beirao2014}.

\subsection{Isogeometric spaces on a surface}\label{sec:IGA-surface}
For substructuring methods, and in order to deal with discrete spaces defined on the interface between two subdomains, we have to make use of isogeometric spaces defined on a surface. Let $\Gamma \subset \partial \Omega$, and we assume that $\Gamma$ is the image through the parametrisation $\map$ of one boundary side of the reference domain $\onref{\Omega}$. Thus, it can be parametrised as $\map_\Gamma:\onref{\Gamma} \rightarrow \Gamma \subset \mathbb{R}^3$. Applying the trace operators defined in Section~\ref{sec:sobolev}, we obtain
\begin{align}
& \traceg{\Gamma}: \Spone{\Omega} \longrightarrow \Spone{\Gamma}, \label{eq:trace1}\\
& \truxng{\Gamma}: \Spone{\Omega} \longrightarrow \Sponestar{\Gamma}, \label{eq:trace2}
\end{align}
where, similarly to \eqref{eq:S1hat}, we first define the spaces in the parametric domain
\begin{align}
& \Spone{\onref{\Gamma}} = S_{p-1,p}(\Xi_k',\Xi_l)\times S_{p,p-1}(\Xi_k,\Xi_l'), \label{eq:Spone-intref}\\
& \Sponestar{\onref{\Gamma}} = S_{p,p-1}(\Xi_k,\Xi_l')\times S_{p-1,p}(\Xi_k',\Xi_l), \label{eq:Sponestar-intref}
\end{align}
and the indices of the knot vectors $1 \le k < l \le 3$ depend on the chosen boundary side. The trace spaces are then defined analogously to \eqref{eq:S1} as
\begin{align}
& \Spone{\Gamma} = \{ \vect{v} : \mapjac_\Gamma^\top (\vect{v} \circ \map_\Gamma) \in \Spone{\onref{\Gamma}} \}, \label{eq:Spone-int} \\
& \Sponestar{\Gamma} = \{ \vect{v} : \left|\mapjac_\Gamma\right| \mapjac_\Gamma^{+} (\vect{v} \circ \map_\Gamma) \in \Sponestar{\onref{\Gamma}} \}, \label{eq:Sponestar-int}
\end{align}
where $\mapjac_\Gamma$ is the Jacobian matrix of $\map_\Gamma$, with size $2 \times 3$, for which we also define the measure $\left|\mapjac_\Gamma\right| = \sqrt{\text{det}(\mapjac_\Gamma^\top \mapjac_\Gamma)}$, and the Moore-Penrose pseudo-inverse $\mapjac_\Gamma^+= (\mapjac_\Gamma^\top \mapjac_\Gamma)^{-1}\mapjac_\Gamma^\top$.

Analogously to the three-dimensional case, these trace spaces are part of a de Rham diagram. In fact, if we introduce the spaces
\[
S^0_p(\onref{\Gamma}) = S_{p,p}(\Xi_k,\Xi_l), \quad S^2_p(\onref{\Gamma}) = S_{p-1, p-1}(\Xi'_k, \Xi'_l), 
\]
and then apply a push-forward with the usual maps for differential forms, we obtain (see \cite{Buffa_2018} for further details)
\begin{equation}\label{eq:2D-sequence-div}
\begin{tikzcd}
\R \arrow[r] &S^{0}_p({\Gamma}) \arrow[r,"\curlG{}"] &S^{1^*}_p({\Gamma}) \arrow[r,"\divG{}"] &S^{2}_p({\Gamma}) \arrow[r] &0,
\end{tikzcd}
\end{equation}
where the subscript $\Gamma$ is used to denote the surface differential operators. Moreover, we can define spaces with vanishing boundary conditions, and in particular functions in $S^0_p(\Gamma;\partial \Gamma)$ vanish on $\partial \Gamma$, while for $S^1_p(\Gamma;\partial \Gamma)$ the tangential component vanishes. Introducing the notation $S^2_p(\Gamma; \partial \Gamma) = S^2_p(\Gamma)$, we also have the exact sequence
\begin{equation} \label{eq:2D-sequence-curl}
\begin{tikzcd}
0 \arrow[r] &S^{0}_p({\Gamma};\partial {\Gamma}) \arrow[r,"{\gradG}"] &S^{1}_p({\Gamma}; \partial {\Gamma}) \arrow[r,"\curlsG{}"] &S^{2}_p({\Gamma};\partial {\Gamma}) \arrow[r] &\R.
\end{tikzcd}
\end{equation}
Obviously, the construction of the spaces on the surface can be generalised to arbitrary degree.

\section{Substructuring}
\label{sec:substructuring}

In this section we present two instances of domain decomposition methods. The aim is twofold. Given the typical structure of \gls{rf} cavities, it is desirable to obtain a substructuring method able to exploit the modularity of the design in order to speed up matrix assembly and eventually reduce memory consumption. On the other hand, we are interested in the flexibility of coupling different discretisations across different domains. A classical Galerkin approximation cannot be straightforwardly applied, since the discrete space consists of discontinuous functions across the connecting interface, hence it is not a subset of $\Hcurl{\Omega}{}$ anymore. The approach is to give a weak formulation of~\eqref{eq:curl-curl} compatible with the independent definition of the finite dimensional spaces on each subdomain, with the addition of a weak coupling condition for the tangential fields across the interfaces. We will refer in the following only to the eigenvalue problem introduced and its weak formulation~\eqref{eq:curl-curl-weak}, although the methods proposed can be applied to source problems as well.

\subsection{The general setting}

Let us assume that the domain is decomposed in $N_{\text{dom}}$ non-overlapping subdomains, such that $\overline \Omega = \bigcup_{i=1}^{N_{\text{dom}}} \overline \Omega_i$. Let us denote $\Gamma_{ij} = \overline{\Omega}_i\cap\overline{\Omega}_j$ for $i<j$, such that $\Omega_i$ will play the role of the \emph{slave subdomain}, and by $\vect{n}_{\Gamma_{ij}}$, or $\vect{n}_\Gamma$ to simplify notation, the unit normal vector to $\Gamma_{ij}$ pointing outward $\Omega_i$ (see Fig.~\ref{fig:two_domains}). We also introduce $\vect{\mathcal{I}}$ the collection of couples $(i,j)$ with $i < j$ such that $\Gamma_{ij}$ is not empty. Moreover, for each $\Omega_s$ we set $\Sigma_s = \partial \Omega \cap \partial \Omega_s$. On each subdomain we introduce the space $V_s = \Hcurl{\Omega_s}{\Sigma_s}$, and define $V = \prod_{s=1}^{N_{\text{dom}}} V_s$. Moreover, for each interface $\Gamma_{ij}$ we introduce the space $M_{ij} = \{\curl \vect{v} \times \vect{n}_{\Gamma_{ij}} : \vect{v} \in V_i\}$, and define $M = \prod_{(i,j) \in \vect{\mathcal{I}}} M_{ij}$. With this notation, the tangential continuity of the electromagnetic field across the interfaces can be enforced weakly by means of a Lagrange multiplier, and the Maxwell eigenvalue problem \eqref{eq:curl-curl-weak} is equivalent to the mixed variational problem: Find $\vect{E} \in V$, $\vect{\lambda} \in M$ and $\omega \in \mathbb{R}^+$ such that
\begin{equation} \label{eq:mixed-cont}
\begin{aligned}
\displaystyle a(\vect{E}, \vect{v}) + b(\vect{v}, \vect{\lambda}) &= \epsilon_0 \mu_0 \omega^2 \sum_{s=1}^{N_\text{dom}} (\vect{E}, \vect{v})_{\Omega_s} && \quad \forall \vect{v} \in V \\
\displaystyle b(\vect{E}, \vect{\mu}) &= 0 && \quad \forall \vect{\mu} \in M,
\end{aligned}
\end{equation}
where we define the bilinear forms
\begin{align*}
& a(\vect{E}, \vect{v}) := \sum_{s=1}^{N_\text{dom}} (\curl \vect{E}, \curl \vect{v})_{\Omega_s}, \\
& b(\vect{v}, \vect{\mu}) := \sum_{(i,j) \in \vect{\mathcal{I}}} ([\traceg{\Gamma_{ij}}(\vect{v})], \vect{\mu})_{\Gamma_{ij}} = \sum_{(i,j) \in \vect{\mathcal{I}}} ([(\vect{n}_{\Gamma_{ij}} \times \vect{v}) \times \vect{n}_{\Gamma_{ij}}], \vect{\mu})_{\Gamma_{ij}},
\end{align*}
and the brackets $[\cdot]$ denote the jump across the interface $\Gamma_{ij}$. We note that in the continuous setting $\vect{\lambda} = \curl \vect{E} \times \vect{n}_\Gamma$, and the tangential components of $\vect{E}$ are continuous.
\begin{figure}
	\centering
	\begin{tikzpicture}
		\node[anchor=south west,inner sep=0] at (0,0) {\includegraphics[width=0.5\textwidth]{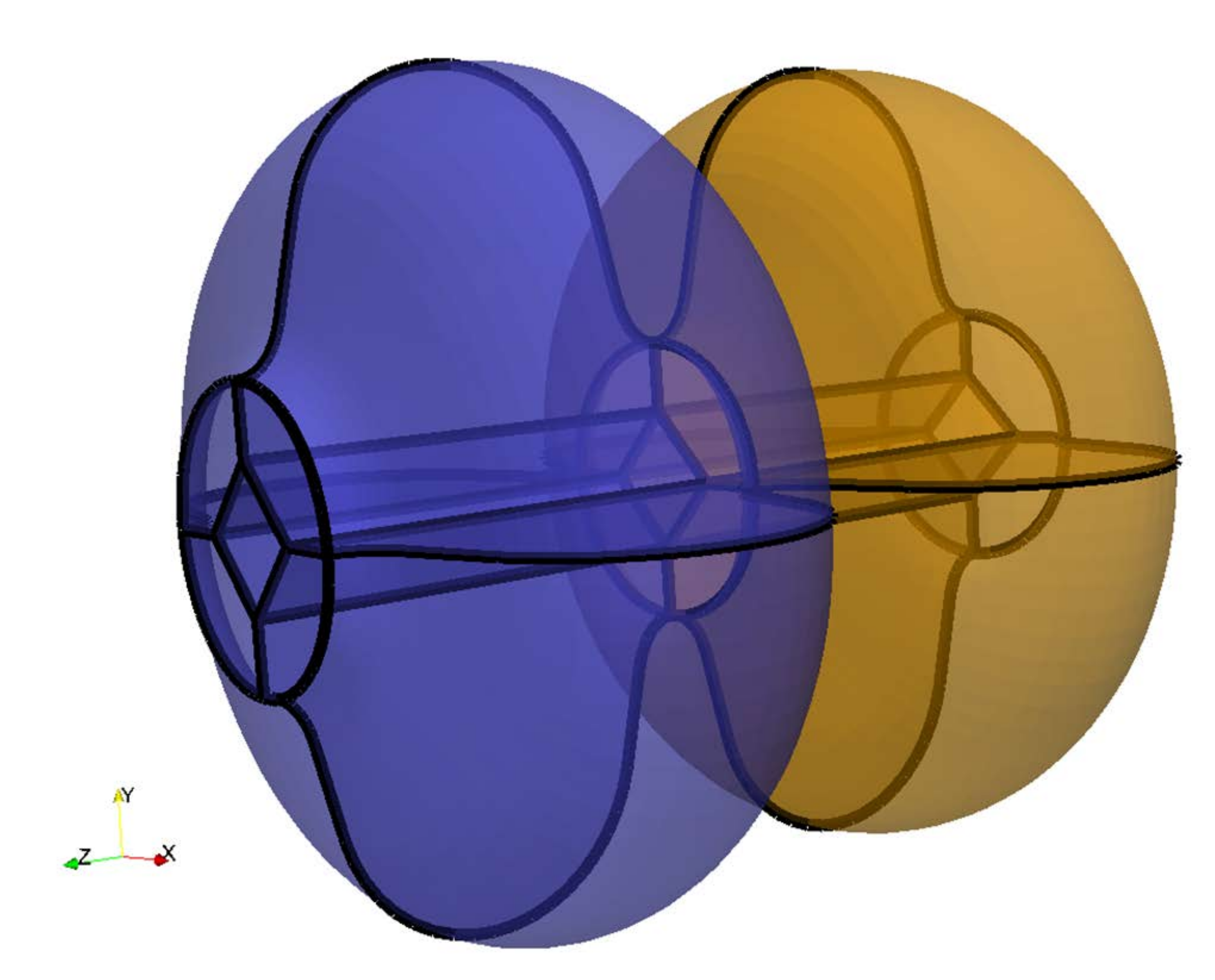}};
		\draw[<-, very thick] (3.8,3) -- (6,0.5);
		\node[anchor=west] at (6,0.5) {$\Gamma_{16}$};
		\draw[<-, very thick] (5.2,3.8) -- (7,2.5);
		\node[anchor=west] at (7,2.5) {$\Omega_6$};
		\draw[<-, very thick] (2,3.5) -- (0.5,4.5);
		\node[anchor=east] at (0.5,4.5) {$\Omega_1$};
	\end{tikzpicture}
	\caption{Schematics of substructuring into two subdomains (blue and yellow). The two subdomains consist of several patches
$\Omega_i$ with $i=1,\ldots,5$ and $\Omega_j$ with $j=6,\ldots,10$, respectively. These patches are glued to each other by the multipatch approach. Lagrange multipliers are used across the subdomain interface, e.g. on $\Gamma_{16}$.\label{fig:two_domains}}
\end{figure}
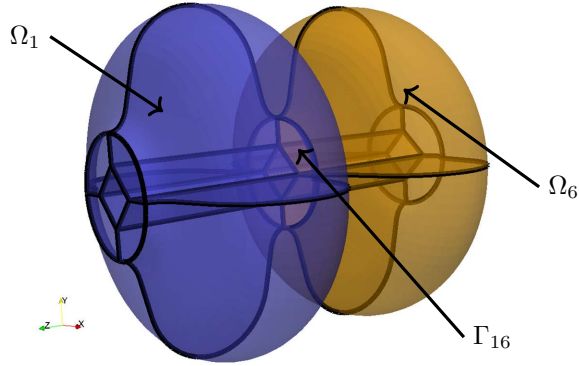

For the discretisation of \eqref{eq:mixed-cont}, let us assume that $N_{\text{dom}} = N_{\text{IGA}} + N_{\text{FEM}}$, and that the first $N_{\text{IGA}}$ of subdomains are defined through a \gls{nurbs} parametrisation and carry an \gls{iga} discretisation, while the last $N_{\text{FEM}}$ ones carry a \gls{fem} discretisation. More precisely, for the \gls{iga} subdomains, i.e. for $s=1, \ldots, N_{\text{IGA}}$, we define the spaces $V_{s,h} = \Spone{\Omega_s;\Sigma_s}$, and with some abuse of notation we let $\Spone{\Omega_s;\Sigma_s} = \Spone{\Omega_s}$ when $\Sigma_s$ is empty. Similarly, for the \gls{fem} subdomains, i.e. for $ s = N_{\text{IGA}} +1, \ldots N_{\text{dom}}$, we denote by $V_{s,h} = \Vfem{\Omega_s;\Sigma_s}$ the corresponding spaces of N\'ed\'elec finite elements with vanishing tangential component on $\Sigma_s$. Then, we define the discrete space as $V_h = \prod_{s=1}^{N_{\text{dom}}} V_{s,h}$. Moreover, for every non-empty interface $\Gamma_{ij}$ we assume for simplicity that the slave subdomain $\Omega_i$ carries an \gls{iga} discretisation, that is, for every couple $(i,j) \in \vect{\mathcal{I}}$, it holds $1 \le i \le N_\text{IGA}$. Then, the Lagrange multiplier is approximated in the discrete space $M_h = \prod_{(i,j) \in \vect{\mathcal{I}}} M_{ij,h}$, where the discrete space on each interface $M_{ij,h}$ will depend on the chosen method, and will be detailed in the next subsections. After we have defined the approximation spaces, we write the discrete weak formulation of the problem: Find $\vect{E}_h \in V_h$, $\vect{\lambda}_h \in M_h$ and $\omega_h \in \mathbb{R}^+$ such that
\begin{equation}\label{eq:saddle-point-weak}
\begin{aligned}
\displaystyle a(\vect{E}_h, \vect{v}_h) + b(\vect{v}_h, \vect{\lambda}_h) &= \epsilon_0 \mu_0 \omega_h^2 \sum_{s=1}^{N_\text{dom}} (\vect{E}_h, \vect{v}_h)_{\Omega_s} && \quad \forall \vect{v}_h \in V_h \\
\displaystyle b(\vect{E}_h, \vect{\mu}_h) &= 0 && \quad \forall \vect{\mu} \in M_h.
\end{aligned}
\end{equation}

\subsection{The State Space Concatenation Method}\label{sec:ssc}
This section discusses the \gls{ssc} method introduced by Flisgen et al.~\cite{Flisgen_2013aa,Flisgen_2015aa} and its properties with respect to standard \gls{dd} methods. It was proposed for the simulation of long chains of resonant cavities as they are present in particle accelerators. In this case, the connections between each resonator is a short waveguide with a circular or rectangular cross-section. The field in these interconnecting parts is assumed to resemble that of a waveguide (the longer the connection, the more this assumption is reasonable) and \gls{ssc} aims at exploiting this a-priori knowledge to choose the Lagrange multiplier. Furthermore, \gls{ssc} applies model order reduction to reduce the dimension of $V_h$ in order to allow for the simulation of very large structures. However, we will only focus here on the domain decomposition aspect.

Although the method arises from physical considerations, the technique proposed can be considered in a more general setting as a domain decomposition method where a modal basis is chosen for the Lagrange multipliers, see the related method in \cite{Deparis_2018aa}. We assume that the interface $\Gamma$ is a simply connected planar surface such that $\partial \Gamma \subset \Sigma$, that in general can be made by several pieces, in the form $\overline{\Gamma} = \bigcup_{(i,j) \in \vect{\mathcal{I}}_\Gamma} \overline \Gamma_{ij}$ with $\vect{\mathcal{I}}_\Gamma \subset \vect{\mathcal{I}}$. We also assume that the interface $\Gamma$ is perpendicular to the boundary, in the sense that, denoting by $\vect{n}$ the unit normal vector exterior to $\Omega$, it holds $\vect{n}_\Gamma \cdot \vect{n} = \vect{0}$. The Lagrangian basis is obtained from the modes that can be excited in a waveguide with $\Gamma$ as the cross-section. These are the \gls{tm} and \gls{te} modes, \cite[Chapter 8.3]{Jackson_1998aa}. Assuming for simplicity of notation $\mathbf{n}_\Gamma=\mathbf{n}_z$, then they are given, up to a multiplicative constant, by $\gradG E_z$ and $\curlG H_z$, where $E_z$ and $H_z$ are solutions of the eigenvalue problems
\begin{equation}\label{eq:eigenmodes-gamma}
\left\{
\begin{aligned}
\Delta_\Gamma E_z + \gamma^2 E_z &= 0 && \text{in } \Gamma\\
E_z &= 0 && \text{on } \partial\Gamma,
\end{aligned}
\right.
\quad
\left\{
\begin{aligned}
\Delta_\Gamma H_z + \gamma^2 H_z &= 0 && \text{in } \Gamma\\
\frac{\partial H_z}{\partial \vect{n}} &= 0 && \text{on } \partial\Gamma,
\end{aligned}
\right.
\end{equation}
with the transverse Laplace operator $\Delta_\Gamma v(x,y) = \laplace v(x,y) - \frac{\partial^2}{\partial z^2} v(x,y)$ on the interface $\Gamma$. Assuming that the boundary of $\Gamma$ is smooth enough, the modes $ \vect{\varphi}_k $, i.e., either $\gradG E_{z,k}$ or $\curlG H_{z,k}$, can be obtained as the numerical or closed-form solution of the eigenvalue problems on $\Gamma$. 
The equations~\eqref{eq:eigenmodes-gamma} have an infinite number of solutions which constitute an orthogonal set of basis functions that can be sorted in ascending order according to their separation constants ($\gamma_1 \leq \gamma_2 \leq \gamma_3 \leq \dots$) and form a basis of $\vect{L}^2(\Gamma)$. This means that for all $\vect{\varphi} \in \vect{L}^2(\Gamma)$ there exist $\{\alpha_k\}_{k \in \mathbb{N}}$ such that $\vect{\varphi} = \sum_{k} \alpha_k \vect{\varphi}_k$.

To obtain the saddle-point formulation \eqref{eq:saddle-point-weak} we introduce a new variable $\vect{\varphi} \in \vect{L}^2(\Gamma)$ and consider the problem on each subdomain:
\begin{equation}\label{eq:ssc-one-subd}
\begin{aligned}
\curl{\left(\muvac^{-1}\curl{\vect{E}_s}\right)} &= \omega^2\epsvac\vect{E}_s    &&\text{in } \Omega_s\\
\vect{E}_s\times\vect{n} &= 0       &&\text{on }\partial\Omega\cap\partial\Omega_s\\
\left(\muvac^{-1}\curl{\vect{E}_s}\right) \times \vect{n}_\Gamma &= \vect{\varphi} = \sum_{k} \alpha_k \vect{\varphi}_k &&\text{on } \Gamma,
\end{aligned}
\end{equation}
that is, we impose on $\Gamma$ a superposition of the waveguide modes as the Neumann data. By translating Eq.~\eqref{eq:ssc-one-subd} in the weak sense and by adding the weak continuity of the solution across $\Gamma$ we obtain \eqref{eq:mixed-cont}.

The discrete space $M_h = \text{span}\lbrace\vect{\varphi}_k, k=1,\dots,\Ngamma\rbrace$ is obtained by truncating the series, selecting the first $\Ngamma<\infty$ modes and we arrive at \eqref{eq:saddle-point-weak}. If the interface is placed in correspondence of a waveguide-like section, it is possible to exploit physical knowledge on the dispersion relation to estimate the number of modes required to obtain a sufficiently good approximation~\cite{Flisgen_2013aa}. However, in general many modes may be necessary for an accurate representation, which may endanger stability.
 
\subsection{The Mortar Method}\label{sec:mortar}
We now propose the new isogeometric mortar method, where we approximate the Lagrange multiplier with a suitable space of splines. To define $M_h$ we will define the space $M_{ij,h}$ associated to each interface separately. By our assumptions, we know that for each interface $\Gamma_{ij}$ the slave subdomain $\Omega_i$ is discretised with \gls{iga}, and $\Gamma_{ij}$ is a full mapped face of that subdomain. Moreover, as we have seen in \eqref{eq:trace1}, the tangential trace operator $\traceg{\Gamma_{ij}}$ maps the isogeometric discrete space defined in $\Omega_i$, $\Spone{\Omega_i;\Sigma_i}$, into the space $\Spone{\Gamma_{ij}}$, and more precisely onto $\Spone{\Gamma_{ij}; \partial \Gamma_{ij} \cap \partial \Omega}$, the space with boundary conditions.

The approximation of the Lagrange multiplier $\vect{\lambda} = \curl \vect{E} \times \vect{n}$ requires a discrete space of divergence conforming splines, and such that the pairing with the image of $\traceg{\Gamma_{ij}}$ satisfies an inf-sup stability condition on $b(\cdot, \cdot)$. A natural candidate would be the discrete space $\Sponestar{\Gamma_{ij}}$, which is the image of the trace operator $\truxng{\Gamma_{ij}}$. Unfortunately, this choice does not satisfy the inf-sup stability condition. 

In our mortar method we choose an analogous space with a different degree. More precisely, introducing the degree $1 \le q < p$, we define the space $M_{ij,h} = \Sonestar{q}{\Gamma_{ij}}$ analogously to \eqref{eq:Sponestar-intref}-\eqref{eq:Sponestar-int}, removing only $p-q$ repetitions of the first and last knots. Notice that the definition requires the following assumption, that implies that $\Spone{\Omega_i;\Sigma_i}$ consists of continuous functions (both the normal and tangential components) when restricted to $\Omega_i$.
\begin{assumption}\label{as:regularity}
The degree satisfies $p \ge 2$, and all the internal knots of $\Xi_k$, for $k = 1, 2, 3$, are repeated at most $q$ times, with $1 \le q < p$.
\end{assumption}

The analysis of the mortar method, that we develop in Section~\ref{sec:analysis}, reveals that our choice of the discrete space is stable for $q = p-1$. Moreover, the numerical tests in Section~\ref{sec:results} show stability for $q = p-k$ when $k$ is odd, while it is unstable when $k$ is even.

\begin{remark}
The construction of the discrete space takes advantage of the high continuity of splines, and does not have an analogue for N\'ed\'elec finite elements. Indeed, Assumption~\ref{as:regularity} implies that both the tangential and the normal components of the functions are continuous, which is not the case in \gls{fem}.
\end{remark}
 
\subsection{Discussion of the two approaches}\label{sec:mortar_vs_ssc}
It is worth noticing that, with respect to the mortar method, the \gls{ssc} has the advantage that the computation of the coupling terms $b(\vect{v}_h, \vect{\mu}_h)$  is completely independent on each side since the Lagrange multipliers $\vect{\varphi}_k$ on the interfaces live on spaces independent of the volume discretisation. When dealing with the coupling of non-conforming meshes across $\Gamma$, the mortar method requires the construction of a common mesh given by the intersection of the meshes on the two sides. In the \gls{ssc} case, however, the coupling matrices can be straightforwardly assembled on completely different meshes. Furthermore, for accelerator cavities with common shapes, the intersection surface $\Gamma$ is usually sufficiently simple that even closed-form solutions for the eigenmodes can be used. However, spectral correctness and inf-sup stability of \gls{ssc} are not guaranteed as we will show in the example section. On the other hand, we show these properties for our mortar method in the next section. Another important difference is that for the \gls{ssc} case the interface must be perpendicular to the boundary, and its boundary must be contained on the boundary of the domain, while the mortar method can deal with arbitrary interfaces.

\section{Analysis of the mortar method}
\label{sec:analysis}

In this section we analyse the mortar method introduced in Section~\ref{sec:mortar}. We start by introducing the conditions for the spectral correctness of any method written in the general setting \eqref{eq:saddle-point-weak}.

\subsection{Conditions for spectral correctness}\label{sec:correctness}
In order to analyse the spectral correctness of the method, we need to define the space of discrete functions satisfying the weak continuity condition, namely
\[
V_{h,M} := \{ \vect{u}_h \in V_h : b(\vect{u}_h,\vect{\mu}_h) = 0 \quad \forall \vect{\mu}_h \in M_h \}.
\]
We also define the discrete kernel as the subspace
\begin{equation}
K_{h,M} := \{ \vect{u}_h \in V_{h,M} : a(\vect{u}_h, \vect{v}_h) = 0 \quad \forall \vect{v}_h \in V_{h,M} \}, \label{eq:KHM}
\end{equation}
and denote by $W_{h,M} = K_{h,M}^\perp$ the orthogonal space to $K_{h,M}$ with respect to the $L^2$ product.

Several (necessary and sufficient) conditions have to be checked to ensure that the solution of \eqref{eq:saddle-point-weak} provides a spectrally correct approximation of \eqref{eq:mixed-cont}:
\begin{property}[Inf-sup stability] \label{property:inf-sup}
There exists a constant $\beta > 0$ and $h_0 > 0$ such that for $h < h_0$ it holds
\begin{equation}
\sup_{\vect{u}_h \in V_h} \frac{b(\vect{u}_h,\vect{\mu}_h)}{\left( \sum_{s=1}^{N_\text{dom}} \| \vect{u}_h\|^2_{0,\curl,\Omega_s} \right)^{1/2} } \ge \beta \bigg( \sum_{(i,j) \in \vect{\mathcal{I}}}\| \vect{\mu}_h \|^2_{-1/2,\div{},\Gamma_{ij}} \bigg)^{1/2} \; \forall \vect{\mu}_h \in M_h. \label{eq:inf-sup}
\end{equation}
\end{property}

\begin{property}[Completeness of the discrete kernel] \label{property:completeness}
For any $\phi \in H^1_0(\Omega)$, we have that
\begin{equation}
\lim_{h\rightarrow 0} \inf_{\vect{u}_h \in K_{h,M} } \| \grad \phi - \vect{u}_h\|_{0} = 0. \label{eq:CDK}
\end{equation}
\end{property}

\begin{property}[Gap property, or discrete compactness (see \cite{Boffi_2010aa} or \cite{Buffa05})] \label{property:gap}
There exist constants $\sigma, C > 0$ such that, for every $\vect{w}_h \in W_{h,M}$ there exists $\vect{w} \in \Hcurl{\Omega}{0} \cap \Hdiv{\Omega}{}$ with $\div \vect{w} = 0$ such that 
\begin{equation*}
\| \vect{w}_h - \vect{w} \|_0 \le C h^\sigma \| \vect{w} \|_{0,\curl{}}. \label{eq:GAP}
\end{equation*}
\end{property}

\subsection{Properties of the mortar method}
In what follows we prove Properties~\ref{property:inf-sup} and~\ref{property:completeness} for the mortar method, while the proof of the gap property, which is more intricate, is left for further studies.

For simplicity, in the following we restrict ourselves to the case of two non-overlapping subdomains, $\overline \Omega = \overline \Omega_1 \cup \overline \Omega_2$ with a common interface $\Gamma = \Gamma_{12} = \partial \Omega_1 \cap \partial \Omega_2$, and assume that both subdomains are discretised with \gls{iga}. The extension to the case of $\Omega_2$ discretised with \gls{fem} is straightforward, and only affects the proof of Proposition~\ref{prop:CDK}. From now on, we assume that $q=p-1$, that is, the Lagrange multiplier belongs to the space $\Sonestar{p-1}{\Gamma}$.

The proofs rely on two known results: the Helmholtz decomposition and the gap property, both at the level of the interface.
\begin{lemma}[Helmholtz decomposition]
The two following Helmholtz decompositions hold:
\[
S^1_p(\Gamma;\partial \Gamma) = Z^1_p \oplus \gradG{} (S^0_p(\Gamma;\partial \Gamma)), \qquad
S^{1^*}_q(\Gamma) = Z^{1^*}_q \oplus \curlG{} (S^0_q(\Gamma) \cap L^2_0(\Gamma)),
\]
with 
\begin{align*}
Z^1_p &= \{ \vect{u} \in S^1_p(\Gamma;\partial \Gamma) : (\vect{u}, \gradG{} \phi) = 0 \quad \forall \phi \in S^0_p(\Gamma;\partial \Gamma) \}, \\
Z^{1^*}_q &= \{ \vect{v} \in S^{1^*}_q(\Gamma) : (\vect{v}, \curlG{} \psi) = 0 \quad \forall \psi \in S^0_q(\Gamma) \},
\end{align*}
and $L^2_0(\Gamma)$ the $L^2$ functions with zero average value.
\end{lemma}
\begin{proof}
The result is a consequence of the commuting projectors defined in \cite{Buffa_2011aa}.
\end{proof}
\begin{lemma}[Gap property on the interface]\label{lemma:gap}
There exists a positive constant $\sigma > 0$ such that, for each $\vect{u}_h \in Z^1_p$,
there exists a function $\vect{w} \in \Hcurl{\Gamma}{0}$ with $(\vect{w}, \gradG{\phi}) = 0$ for all $\phi \in H^1_0(\Gamma)$, satisfying
\[
\curlsG \vect{u}_h = \curlsG \vect{w}, \quad \| \vect{u}_h - \vect{w} \|_0 \lesssim h^\sigma \| \curlsG \vect{u}_h \|_{-1/2}.
\]
Similarly, for each $\vect{v}_h \in Z^{1^*}_q$,
there exists a function $\vect{z} \in \Hdiv{\Gamma}{}$ 
with $(\vect{z}, \curlG \psi) = 0$ for all $\psi \in H^1(\Gamma)$, satisfying
\[
\divG \vect{v}_h = \divG \vect{z}, \quad \| \vect{v}_h - \vect{z} \|_0 \lesssim h^\sigma \| \divG \vect{v}_h \|_{-1/2}.
\]
Moreover, $(\vect{w}, \vect{z}) = 0$.
\end{lemma}
\begin{proof}
The proof of the existence, and of the two inequalities, is an adaption to the two-dimensional case of Lemma~6.1 in \cite{Buffa_2011aa}, using in particular \cite[Theorem~4.1]{HIP02a} and \cite[Lemma~2.3]{Hiptmair-Schwab} (see also \cite[Lemma~6.2]{Hiptmair-Schwab}).
The orthogonality is a consequence of Helmholtz decomposition, because $\divG \vect{w} = 0$ and $\curlsG \vect{z} = 0$.
\end{proof}
As a consequence of the gap property, we have the following discrete Friedrichs' inequalities.
\begin{corollary}
For ${\bf u}_h \in Z_p^1$ and for ${\bf v}_h \in Z_q^{1^*}$, it holds
\begin{equation*}
\| {\bf u}_h \|_0 \lesssim \| \curlsG{{\bf u}_h}  \|_{-1/2} , \qquad \| {\bf v}_h \|_0 \lesssim \| \divG{{\bf v}_h} \|_{-1/2}.
\end{equation*}
\end{corollary}

\subsubsection{Proof of the inf-sup condition}
The inf-sup condition \eqref{eq:inf-sup} is a consequence of the following proposition, in particular \eqref{eq:inf-sup1}, and the continuity of the trace operators.
\begin{proposition}[inf-sup condition]\label{prop:inf-sup}
Let the spaces $S_p^i(\Gamma;\partial \Gamma)$ for $i=0,1,2$ be defined as in \eqref{eq:2D-sequence-curl}, starting from the knot vectors $\Xi_1, \Xi_2$, and let Assumption~\ref{as:regularity} hold. Let the spaces $S_{p-1}^i(\Gamma)$ for $i = 0, 1^*, 2$ be defined as in \eqref{eq:2D-sequence-div}, starting from knot vectors $\Xi_1', \Xi_2'$. Then, there exist $\beta_0, \beta_1, \beta_2 > 0$ and $h_0 > 0$ such that the following inf-sup conditions hold for $s \in [0,1]$ and for all $h < h_0$:
\begin{subequations}
\begin{align}
\sup_{u \in S_p^0(\Gamma;\partial \Gamma)} \frac{\int_\Gamma u v}{\|u\|_{H^s}} &\ge \beta_0 \|v\|_{H^{-s}} &&\quad \forall {v \in S_{p-1}^2(\Gamma)}, \label{eq:inf-sup0}\\
\sup_{\vect{u} \in \Spone{\Gamma;\partial \Gamma}} \frac{\int_\Gamma \vect{u} \cdot \vect{v}}{\|\vect{u}\|_{-1/2 , \curls{}}} &\ge \beta_1 \|\vect{v}\|_{-1/2,\div{}} &&\quad \forall {\vect{v} \in \Sonestar{p-1}{\Gamma}}, \label{eq:inf-sup1} \\
\sup_{u \in S_p^2(\Gamma;\partial \Gamma)} \frac{\int_\Gamma u v}{\|u\|_{H^{-s}}} &\ge \beta_2 \|v\|_{H^s} &&\quad \forall {v \in S_{p-1}^0(\Gamma)}. \label{eq:inf-sup2}
\end{align}
\end{subequations}
\end{proposition}
\begin{proof}
The condition \eqref{eq:inf-sup0} was already proved in \cite[Thm.~12]{Brivadis_2015aa} for $s=0$ and in \cite[Thm.~3.6]{Antolin_2018} for $s \in [0,1]$. The inf-sup condition \eqref{eq:inf-sup2} for $s=0$ is trivial, since both spaces are the same. The condition for $s \in (0,1]$ is proved as in \cite{Antolin_2018}, defining a Fortin operator with the help of the commutative projectors in \cite{Buffa_2011aa}, noting that since the two spaces are equal, we can exchange their roles in the inf-sup condition.

Now let $\vect{v} \in \Sonestar{p-1}{\Gamma}$, that using Helmholtz decomposition we can write as $\vect{v} = \curlG \psi + \vect{v}_{\bot}$ with $\psi \in S^0_{p-1}({\Gamma}) \cap L^2_0({\Gamma})$ and $\vect{v}_\bot \in Z^{1^*}_{p-1}$. We have to find some $\vect{u} \in \Spone{\Gamma;\partial \Gamma}$ such that the inequality holds. First, we apply Helmholtz decomposition to write $\vect{u} = \gradG{\phi} + \vect{u}_\bot$, with $\phi \in S^0_p({\Gamma};\partial {\Gamma})$ and $\vect{u}_\bot \in Z^1_p$. Thanks to the discrete Friedrichs' inequalities, and the zero average value for $\psi$, we can work with the equivalent norms
\[
\| \vect{v} \|_{-1/2,\div{}} \simeq \| \psi\|_{1/2} + \| \divG{\vect{v}_\bot}\|_{-1/2} , \qquad
\| \vect{u} \|_{-1/2,\curls{}} \simeq \| \phi\|_{1/2} + \| \curlsG{\vect{u}_\bot}\|_{-1/2}.
\]
It is readily seen that
\[
\int_{{\Gamma}} \vect{u} \cdot \vect{v} = 
- \int_{{\Gamma}} \phi \divG \vect{v}_\bot + \int_{{\Gamma}} \psi \curlsG \vect{u}_\bot + \int_{{\Gamma}} \vect{u}_\bot \cdot \vect{v}_\bot,
\]
and thanks to Lemma~\ref{lemma:gap}, there exist $\sigma > 0$, $\vect{z} \in \Hdiv{{\Gamma}}{}$ and $\vect{w} \in \Hcurl{{\Gamma}}{0}$ such that
\[
\begin{array}{c}
\displaystyle
\left | \int_{{\Gamma}} \vect{u}_\bot \cdot \vect{v}_\bot \right | = \left| \int_{{\Gamma}} (\vect{u}_\bot - \vect{w}) \cdot \vect{z} + (\vect{v}_\bot - \vect{z}) \cdot \vect{w} + (\vect{u}_\bot - \vect{w}) \cdot (\vect{v}_\bot - \vect{z}) \right| \\
\displaystyle
\lesssim h^\sigma \| \curlsG{\vect{u}_\bot}\|_{-1/2}\| \divG{\vect{v}_\bot}\|_{-1/2}.
\end{array}
\]
Using \eqref{eq:inf-sup0} and \eqref{eq:inf-sup2}, we can choose $\phi \in S_p^0({\Gamma};\partial {\Gamma})$ and $\vect{u}_\bot \in S_p^1({\Gamma};\partial {\Gamma})$ such that
\[
\int_{{\Gamma}} \phi \divG{\vect{v}_\bot} \ge \beta_0 \| \phi \|_{1/2} \| \divG{\vect{v}_\bot} \|_{-1/2}, \quad 
\int_{{\Gamma}} \psi \curlsG{\vect{u}_\bot} \ge \beta_2 \| \curlsG{\vect{u}_\bot} \|_{-1/2} \| \psi \|_{1/2}.
\]
Moreover, the inf-sup condition makes these choices continuous, and we have
\[
\| \phi \|_{1/2} \simeq \| \divG{\vect{v}_\bot} \|_{-1/2}, \qquad \| \curlsG{\vect{u}_\bot} \|_{-1/2} \simeq \| \psi \|_{1/2},
\]
and gathering these last four results, we obtain \eqref{eq:inf-sup1}.
\end{proof}

\subsubsection{Completeness of the discrete kernel}
\begin{proposition} \label{prop:CDK}
Let $M_h = \Sonestar{p-1}{\Gamma}$, and $K_{h,M}$ be defined as in \eqref{eq:KHM}. Then Property~\ref{property:completeness} holds.
\end{proposition}
\begin{proof}
We start characterizing the space $K_{h,M}$. For $\vect{u}_h \in K_{h,M}$, by the definition of $a(\cdot, \cdot)$ it holds that $\curl \vect{u}_h|_{\Omega_k} = \vect{0}$, for $k = 1,2$, and there exists $\phi_{k,h} \in S^0_p(\Omega_k;\Sigma_k)$ such that $\vect{u}_h|_{\Omega_k} = \grad \phi_{k,h}$. Moreover, from the definition of $K_{h,M}$ (and $V_{h,M}$), applying the definition of the surface gradient, the fact that differential and trace operators commute, and integration by parts, we know that
\begin{align*}
b(\vect{u}_h,\vect{\mu}_h) &= \int_\Gamma \gradG (\traceh_{\Gamma} (\phi_{1,h}) - \traceh_{\Gamma} (\phi_{2,h})) \cdot \vect{\mu}_h \\
&= - \int_\Gamma (\traceh_{\Gamma}(\phi_{1,h}) - \traceh_{\Gamma}(\phi_{2,h})) \cdot \divG \vect{\mu}_h = 0 \quad \forall \vect{\mu}_h \in M_h.
\end{align*}
Let us now focus on the mortar constraint. By construction, we knot that $\divG \vect{\mu}_h \in S_{p-1}^2(\Gamma)$, and using the results in \cite{Antolin_2018} we know that there exists a Fortin projector $\Pi_F: L^2(\Gamma) \longrightarrow S^0_p(\Gamma;\partial \Gamma)$, based on the pairing of the two spaces in \eqref{eq:inf-sup0}, such that 
\[
\| \Pi_F \varphi \|_s \leq \| \varphi \|_s \quad \text{ for } s \in [0,1], \, \varphi \in H^s(\Gamma).
\]
Moreover, we denote by $\Pi^0_{\Omega_k}$ and $\Pi^1_{\Omega_k}$, for $k = 1, 2$, the commutative projectors introduced in \cite{Buffa_2011aa} into the spaces $S^0_p(\Omega_k;\Sigma_k)$ and $\Spone{\Omega_k;\Sigma_k}$, respectively. With some abuse of notation, we will also denote $\Pi^0_{\Omega_k} \phi \equiv \Pi^0_{\Omega_k} (\phi|_{\Omega_k})$.

Given $\phi \in H^1_0(\Omega)$, in order to prove \eqref{eq:CDK} we construct $\vect{u}_h \in K_{h,M}$ in the following way:
\[
\vect{u}_h := \left \{
\begin{array}{ll}
\grad (\Pi^0_{\Omega_1} \phi) - \grad {\cal R}_{\Omega_1} \left( \Pi_F \left( \traceh_\Gamma(\Pi^0_{\Omega_1} \phi) - \traceh_\Gamma(\Pi^0_{\Omega_2} \phi) \right) \right) & \text{ in } \Omega_1,\\
\grad (\Pi^0_{\Omega_2} \phi) & \text{ in } \Omega_2,
\end{array}
\right.
\]
where ${\cal R}_{\Omega_1}$ is a continuous extension operator into $S_p^0(\Omega_1; \Sigma_1)$ that solves the discrete Laplacian in $\Omega_1$ imposing a Dirichlet condition on its boundary. To prove the result, we need to show that $\vect{u}_h \in K_{h,M}$, and that it converges to $\grad \phi$.

From the definition of $\vect{u}_h$, it is obvious that its restriction to $\Omega_k$ belongs to $S^1_p(\Omega_k;\Sigma_k)$ and is irrotational, hence $a(\vect{u}_h, \vect{v}_h) = 0$ for any $\vect{v}_h \in V_h$. The second condition to belong to $K_{h,M}$, namely $b(\vect{u}_h, \vect{\mu}_h) = 0$ for all $\vect{\mu_h} \in M_h$ also holds. Indeed, subsequently applying the commutativity of the trace and differential operators, the fact that $\traceh_\Gamma \circ {\cal R}_{\Omega_1}$ is equal to the identity, integration by parts, and that $\Pi_F$ is a projector, we obtain
\begin{align*}
& b(\vect{u}_h, \vect{\mu}_h) = \int_\Gamma [\traceg{\Gamma} \vect{u}_h] \cdot \vect{\mu}_h = \int_\Gamma \Big( \gradG(\traceh_\Gamma(\Pi^0_{\Omega_1} \phi) - \traceh_\Gamma(\Pi^0_{\Omega_2} \phi)) \\
& \; - \gradG{\left(\traceh_\Gamma \left({\cal R}_{\Omega_1} \left( \Pi_F \left( \traceh_\Gamma(\Pi^0_{\Omega_1} \phi) - \traceh_\Gamma(\Pi^0_{\Omega_2} \phi) \right) \right) \right) \right) \Big) } \cdot \vect{\mu}_h \\
& \; = - \int_\Gamma \Big( \traceh_\Gamma(\Pi^0_{\Omega_1} \phi) - \traceh_\Gamma(\Pi^0_{\Omega_2} \phi) -  \Pi_F \left( \traceh_\Gamma(\Pi^0_{\Omega_1} \phi) - \traceh_\Gamma(\Pi^0_{\Omega_2} \phi) \right) \Big) \cdot \divG \vect{\mu}_h = 0,
\end{align*}
which proves that $\vect{u}_h \in K_{h,M}$.

Regarding the convergence, by the commutativity of the $\Pi^0_{\Omega_k}$ and $\Pi^1_{\Omega_k}$ projectors with the gradient, we have
\begin{align*}
\| \grad \phi - \vect{u}_h\|_0 \le & \sum_{k=1}^2 \| \grad \phi - \Pi^1_{\Omega_k} (\grad \phi) \|_{0,\Omega_k} \\
& + \| \grad {\cal R}_{\Omega_1} \left( \Pi_F \left( \traceh_\Gamma(\Pi^0_{\Omega_1} \phi) - \traceh_\Gamma(\Pi^0_{\Omega_2} \phi) \right) \right) \|_{0,\Omega_1}.
\end{align*}
The first term in the sum converges to zero when $h$ tends to zero by the results in \cite{Buffa_2011aa}. For the second term, we use the definition of the continuous extension operator ${\cal R}_{\Omega_1}$, the continuity of the Fortin projector $\Pi_F$ and the trace operator $\traceh_\Gamma$, and the triangular inequality to obtain
\begin{align*}
& \| \grad {\cal R}_{\Omega_1} \left( \Pi_F \left( \traceh_\Gamma(\Pi^0_{\Omega_1} \phi) - \traceh_\Gamma(\Pi^0_{\Omega_2} \phi) \right) \right) \|_{0,\Omega_1} \\
& \le C \| \Pi_F \left( \traceh_\Gamma(\Pi^0_{\Omega_1} \phi) - \traceh_\Gamma(\Pi^0_{\Omega_2} \phi) \right)\|_{1/2,\Gamma} \\
& \le C \| \traceh_\Gamma(\Pi^0_{\Omega_1} \phi - \phi) - \traceh_\Gamma(\Pi^0_{\Omega_2} \phi - \phi)\|_{1/2,\Gamma} \\
& \le C (\| \Pi^0_{\Omega_1} \phi - \phi \|_{0,\Omega_1} + \| \Pi^0_{\Omega_2} \phi - \phi \|_{0,\Omega_2}),
\end{align*}
where $C$ denotes a generic constant independent of $h$. Applying again the results in \cite{Buffa_2011aa}, this term also converges to zero, which finishes the proof.
\end{proof}

\section{Results}
\label{sec:results}

This section presents some numerical results of the applicability of the two substructuring methods introduced in Section~\ref{sec:substructuring}.

\begin{figure}
	\begin{minipage}[t]{.42\textwidth}
		\includegraphics[width=\textwidth]{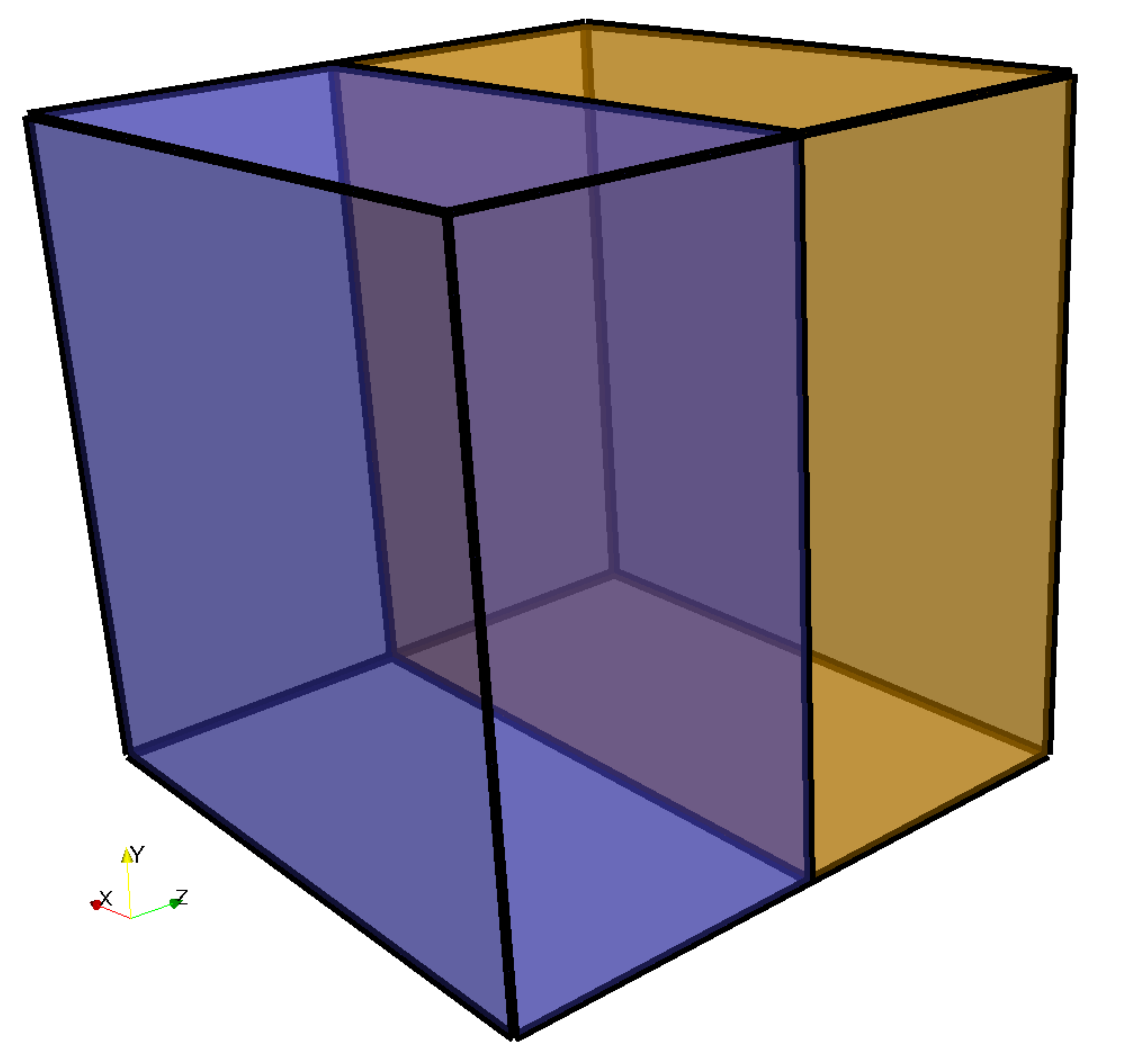}
		\caption[Two conforming patches of the unit cube]{Two conforming patches $\Omega_1$ (in blue) and $\Omega_2$ (in orange) of the unit cube.\label{fig:2cubes_geo}}
	\end{minipage}\hfill
	\begin{minipage}[t]{.45\textwidth}
		\includegraphics[width=\textwidth]{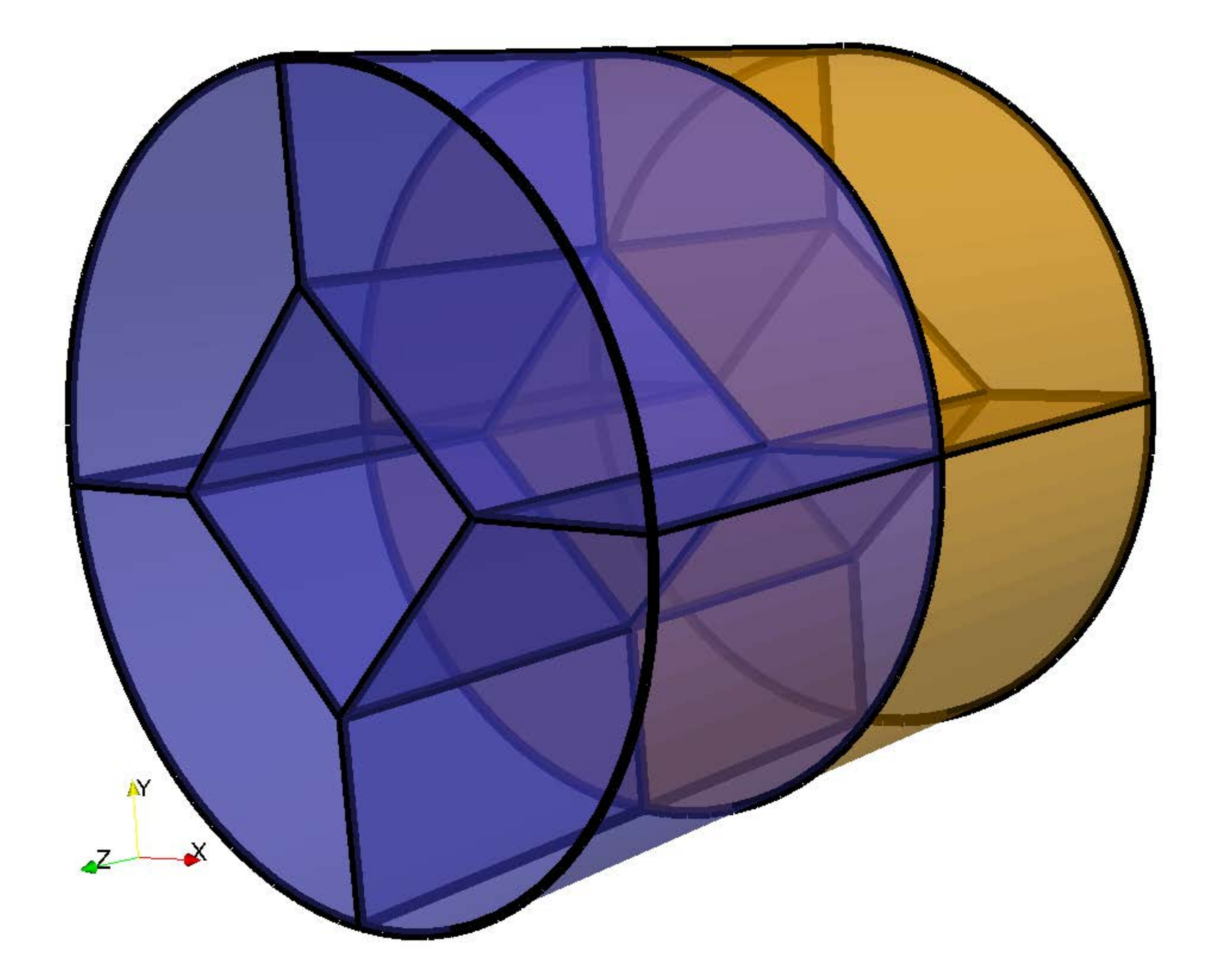}
\caption[Conforming multipatch case of the pillbox]{Conforming multipatch case of the pillbox with domain $\Omega_1$ (in blue) and $\Omega_2$ (in orange).\label{fig:2pillbox_geo}}
	\end{minipage}
\end{figure}

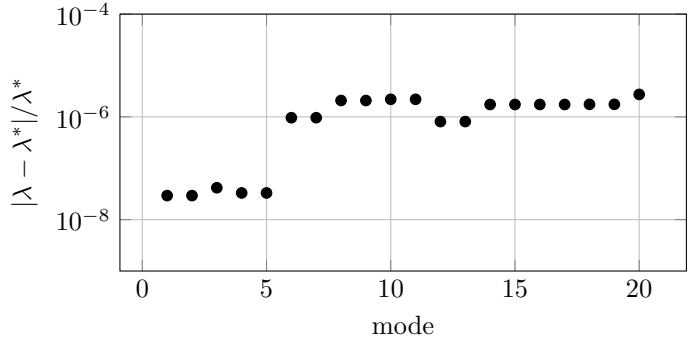
\begin{figure}[ht]
\centering
	\begin{tikzpicture}
	\begin{semilogyaxis}[width=.7\textwidth,height=5cm,xlabel={mode},ylabel={$|\lambda - \lambda^*| / \lambda^*$},ymin=1e-9,ymax=1e-4,grid=major]
	\addplot[mark=*,only marks] table [x={mode}, y={err}, col sep=comma] {fig07.csv};
	\end{semilogyaxis}
	\end{tikzpicture}
	\caption{Mortaring of two single-patch unit cube domains with non-conforming meshes: relative error of the first \num{20} eigenvalues with $p_{\text{IGA}}=4$, $p_{\text{FEM}}=3$, $q=3$, $N_{\text{dof}}=\num{21160}$.\label{fig:2cubes_spectrum}}
\end{figure}

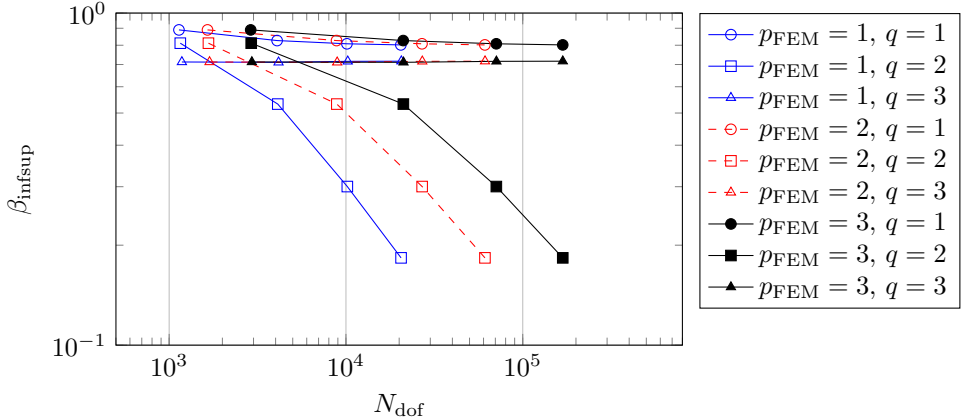
\begin{figure}
	\centering
	\begin{tikzpicture}
	\begin{loglogaxis}[width=.7\textwidth,height=6cm,xlabel={$N_\text{dof}$},ylabel={$\betainfsup$},xmin=500,xmax=8e5,ymin=1e-1,ymax=1,legend pos=outer north east,grid=major]
	\addplot[color=blue, mark=o, solid] table [x={ndof11}, y={infsup11}, col sep=comma] {fig08.csv};
	\addplot[color=blue, mark=square, solid] table [x={ndof12}, y={infsup12}, col sep=comma] {fig08.csv};
	\addplot[color=blue, mark=triangle, solid] table [x={ndof13}, y={infsup13}, col sep=comma] {fig08.csv};
	\addplot[color=red, mark=o, mark options={solid},dashed] table [x={ndof21}, y={infsup21}, col sep=comma] {fig08.csv};
	\addplot[color=red, mark=square, mark options={solid},dashed] table [x={ndof22}, y={infsup22}, col sep=comma] {fig08.csv};
	\addplot[color=red, mark=triangle, mark options={solid},dashed] table [x={ndof23}, y={infsup23}, col sep=comma] {fig08.csv};
	\addplot[color=black, mark=*, solid] table [x={ndof31}, y={infsup31}, col sep=comma] {fig08.csv};
	\addplot[color=black, mark=square*, solid] table [x={ndof32}, y={infsup32}, col sep=comma] {fig08.csv};
	\addplot[color=black, mark=triangle*, solid] table [x={ndof33}, y={infsup33}, col sep=comma] {fig08.csv};
	\legend{{$p_{\text{FEM}}=1,\,q=1$},{$p_{\text{FEM}}=1,\,q=2$},{$p_{\text{FEM}}=1,\,q=3$},
		{$p_{\text{FEM}}=2,\,q=1$},{$p_{\text{FEM}}=2,\,q=2$},{$p_{\text{FEM}}=2,\,q=3$},
		{$p_{\text{FEM}}=3,\,q=1$},{$p_{\text{FEM}}=3,\,q=2$},{$p_{\text{FEM}}=3,\,q=3$}};
	\end{loglogaxis}
	\end{tikzpicture}
	\caption{Mortaring of two single-patch unit cube domains with non-conforming meshes: $\betainfsup$ constants for different space choices. If $q=p_{\text{IGA}}-2$, $\betainfsup$ goes to zero.\label{fig:2cubes_beta}}
\end{figure}

\subsection{Mortar Method}
\paragraph{Cube with two patches} The first test we report is a single patch-to-patch coupling with a trivial geometrical mapping. The unit cube domain $\Omega$ is split in half along the $z$ direction into $\Omega_1$ and $\Omega_2$ (see Fig.~\ref{fig:2cubes_geo}). The coupling interface $\Gamma$ is the square $ \lbrace 0<x<1,\,0<y<1,\,z=1/2\rbrace $. Maxwell's eigenvalue problem~\eqref{eq:curl-curl} is solved using the mortar approach described in section~\ref{sec:mortar}. In $\Omega_1$ we chose an \gls{iga} curl-conforming discretisation with degree $p_{\text{IGA}}=4$ and high regularity $r_{\text{IGA}}=3$, i.e. $V_h = S^1_4(\Omega_1;\Sigma_1)$. On the domain $\Omega_2$ a \gls{fem} discretisation with N\'ed\'elec type hexahedral elements is used with degree $p_{\text{FEM}} = 1,2,3$. This can be straightforwardly accomplished by constructing a B-spline space on $ \Omega_2 $ while setting the regularity of the basis functions to $r_{\text{FEM}}=0$. The grids on the two sides are chosen in such a way that they do not match for any refinement. 

On the interface $\Gamma$ we build the space of Lagrange multipliers $ S^{1^*}_q(\Gamma) $ with degree $q=1,2,3$ and regularity $r=q-1$. The mesh used for the quadrature is given by the intersection of the meshes on both sides which is easy to compute given the tensor product nature of the \gls{iga} hexahedral grid.

The relative errors of the first 20 eigenvalues with respect to the closed form solution are below $10^{-5}$ for the stable case of $p_{\text{FEM}}=3$, $q=3$ and with a total number of degrees of freedom $\Ndof=\num{21160}$ (see Fig.~\ref{fig:2cubes_spectrum}), which shows that there are no spurious eigenvalues. To validate the inf-sup stability we evaluate the inf-sup constant $\betainfsup$ numerically \cite{Chapelle_1993aa} while increasing the mesh refinement level. Fig.~\ref{fig:2cubes_beta} confirms the stability properties expected from the theory: the method is inf-sup stable when we choose $q = p_{\text{IGA}}-1 = 3$. Moreover, we see that the method is also stable for $q=1$, while it is unstable for $q=2$. In general, the inf-sup stability is obtained if $q = p_{\text{IGA}} -k$ with $k$ odd, while it is unstable when $k$ is even.

\paragraph{Pillbox cavity}
We then extend our testing to the case of multipatch geometries. In order to consider a non-trivial mapping, the same test is performed on a cylindrical cavity of radius $R=1$ and length $L=2$ filled with vacuum. The geometry is described with ten NURBS patches of degree 2, (see Fig.~\ref{fig:2pillbox_geo}). As before, we split the cavity in two subdomains separated by the interface $\Gamma = \{(x,y,z) : x^2 + y^2 < 1, z = 1 \}$, and use an \gls{iga} discretisation of degree 4 in $\Omega_1$, and a \gls{fem} discretisation with different degrees in $\Omega_2$. The discretisation spaces on both sides are constructed following the classical multipatch approach such that degrees of freedom lying on adjacent interfaces are glued together. Instead, the Lagrangian multipliers basis is built independently on each of the patches that belong to $\partial \Omega_1 \cap \Gamma$, and that fully describe the interface $ \Gamma $, and the full discrete space $M_h$ is obtained by the union of all of them without any constraint on the connecting lines, i.e. the basis can present jumps across the patches on $ \Gamma $.

In Fig.~\ref{fig:pillbox_beta} the results for the inf-sup constant are shown, for different values of the degree for the \gls{fem} spaces and the multiplier; the behaviour matches the one of the single-patch coupling and the expected one. We also report in Table~\ref{tab:SSC_pillbox_result} the computed eigenfrequencies for $p_{\text{IGA}}=4$, $q=3$ and $p_{\text{FEM}}=1$ in the second mesh ($\Ndof \approx 16000$), along with the exact values, which confirms that no spurious eigenvalues appear.

\begin{figure}
\centering
\begin{tikzpicture}
\begin{loglogaxis}[width=.7\textwidth,height=6cm,,xlabel={$N_\text{dof}$},ylabel={$\betainfsup$},xmin=500,xmax=8e5,ymin=1e-1,ymax=1,legend pos=outer north east,grid=major]
  \addplot[color=blue, mark=o, solid] table [x={ndof11}, y={infsup11}, col sep=comma] {fig09.csv};
  \addplot[color=blue, mark=square, solid] table [x={ndof12}, y={infsup12}, col sep=comma] {fig09.csv};
  \addplot[color=blue, mark=triangle, solid] table [x={ndof13}, y={infsup13}, col sep=comma] {fig09.csv};
  \addplot[color=red, mark=o, mark options={solid},dashed] table [x={ndof21}, y={infsup21}, col sep=comma] {fig09.csv};
  \addplot[color=red, mark=square, mark options={solid},dashed] table [x={ndof22}, y={infsup22}, col sep=comma] {fig09.csv};
  \addplot[color=red, mark=triangle, mark options={solid},dashed] table [x={ndof23}, y={infsup23}, col sep=comma] {fig09.csv};
  \addplot[color=black, mark=*, solid] table [x={ndof31}, y={infsup31}, col sep=comma] {fig09.csv};
  \addplot[color=black, mark=square*, solid] table [x={ndof32}, y={infsup32}, col sep=comma] {fig09.csv};
  \addplot[color=black, mark=triangle*, solid] table [x={ndof33}, y={infsup33}, col sep=comma] {fig09.csv};
  \legend{{$p_{\text{FEM}}=1,\,q=1$},{$p_{\text{FEM}}=1,\,q=2$},{$p_{\text{FEM}}=1,\,q=3$},
          {$p_{\text{FEM}}=2,\,q=1$},{$p_{\text{FEM}}=2,\,q=2$},{$p_{\text{FEM}}=2,\,q=3$},
          {$p_{\text{FEM}}=3,\,q=1$},{$p_{\text{FEM}}=3,\,q=2$},{$p_{\text{FEM}}=3,\,q=3$}};
\end{loglogaxis}
\end{tikzpicture}
\caption{Mortaring of two conforming multipatch pill-box domains with non-conforming meshes: $\betainfsup$ constants for different space choices. If $q=p_{\text{IGA}}-2$, $\betainfsup$ goes to zero.\label{fig:pillbox_beta}}
\end{figure}
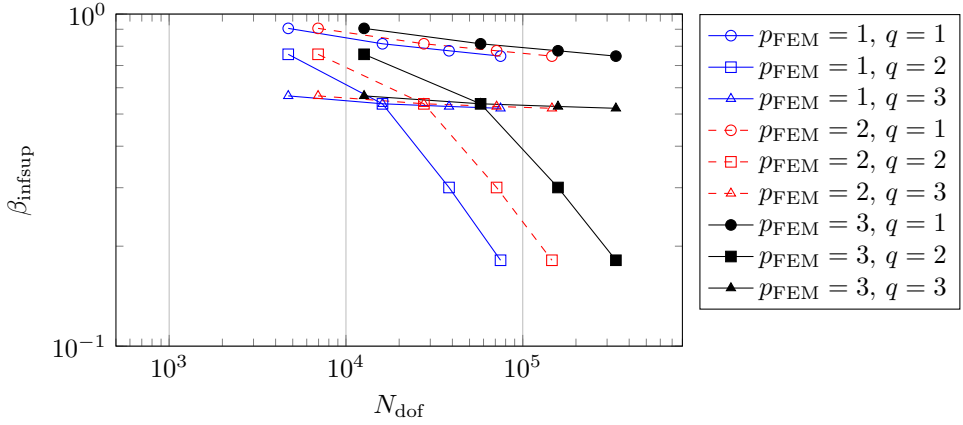

\paragraph{Cube with non-conforming patches}
Finally, we consider again the unit cube, but we further split the subdomain $\Omega_1$ along the $x$ direction into three patches of equal size, while $\Omega_2$ is left unchanged, which gives a geometry described with four non-conforming patches, and with non-conforming meshes. The construction of the discrete spaces on each side, and for the Lagrange multiplier follows along the same lines that for the pillbox cavity, the fact that the patches are not conforming does not pose any difficulty, since the multiplier is defined on each patch separately.

The behaviour of the inf-sup constant, presented in Fig.~\ref{fig:beta_non_conforming}, is analogous to the previous cases. Moreover, the convergence of the 10th eigenvalue, that we report in Fig.~\ref{fig:non-conforming-convergence}, shows that the order of convergence is dominated by the lowest degree of the different discretisation spaces.

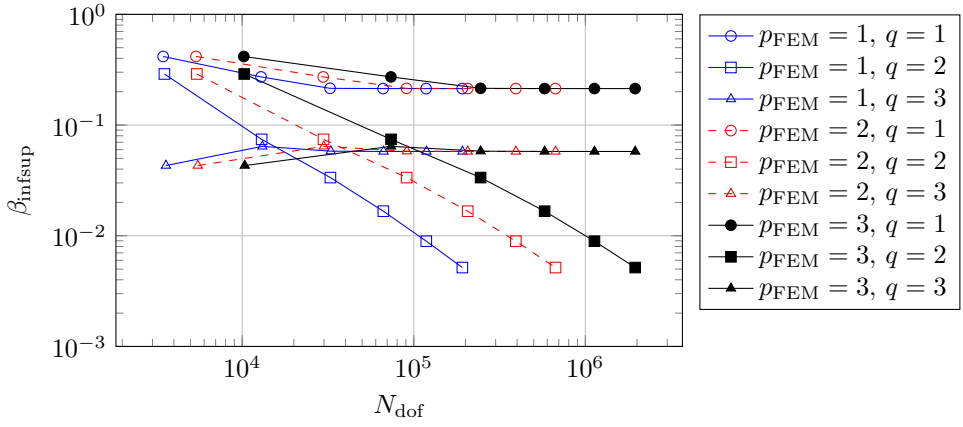
\begin{figure}
\centering
	\begin{tikzpicture}
	\begin{loglogaxis}[width=.7\textwidth,height=6cm,xlabel={$N_\text{dof}$},ylabel={$\betainfsup$},ymin=1e-3,ymax=1,legend pos=outer north east,grid=major]
	\addplot[color=blue, mark=o, solid] table [x={ndof11}, y={infsup11}, col sep=comma] {fig10.csv};
	\addplot[color=blue, mark=square, solid] table [x={ndof12}, y={infsup12}, col sep=comma] {fig10.csv};
	\addplot[color=blue, mark=triangle, solid] table [x={ndof13}, y={infsup13}, col sep=comma] {fig10.csv};
	\addplot[color=red, mark=o, mark options={solid},dashed]  table [x={ndof21}, y={infsup21}, col sep=comma] {fig10.csv};
	\addplot[color=red, mark=square, mark options={solid},dashed] table [x={ndof22}, y={infsup22}, col sep=comma] {fig10.csv};
	\addplot[color=red, mark=triangle, mark options={solid},dashed] table [x={ndof23}, y={infsup23}, col sep=comma] {fig10.csv};
	\addplot[color=black, mark=*, solid] table [x={ndof31}, y={infsup31}, col sep=comma] {fig10.csv};
	\addplot[color=black, mark=square*, solid] table [x={ndof32}, y={infsup32}, col sep=comma] {fig10.csv};
	\addplot[color=black, mark=triangle*, solid] table [x={ndof33}, y={infsup33}, col sep=comma] {fig10.csv};
	\legend{{$p_{\text{FEM}}=1,\,q=1$},{$p_{\text{FEM}}=1,\,q=2$},{$p_{\text{FEM}}=1,\,q=3$},
		{$p_{\text{FEM}}=2,\,q=1$},{$p_{\text{FEM}}=2,\,q=2$},{$p_{\text{FEM}}=2,\,q=3$},
		{$p_{\text{FEM}}=3,\,q=1$},{$p_{\text{FEM}}=3,\,q=2$},{$p_{\text{FEM}}=3,\,q=3$}};
	\end{loglogaxis}
	\end{tikzpicture}
	\caption{Mortaring of two non-conforming multipatch unit cube domains: $\betainfsup$ constants for different space choices. If $q=p_{\text{IGA}}-2$, $\betainfsup$ goes to zero.\label{fig:beta_non_conforming}}
\end{figure}

\begin{figure}[ht]
	\centering
	\begin{tikzpicture}
	\begin{loglogaxis}[width=.7\textwidth,height=6cm,xlabel={$N_\text{dof}$},ylabel={$|\lambda_1 - \lambda_1^*| / \lambda_1^*$},legend pos=outer north east,grid=major]
	\addplot[color=blue, mark=o, solid] table [x={ndof11}, y={error11}, col sep=comma] {fig11.csv};
	\addplot[color=blue, mark=square, solid] table [x={ndof12}, y={error12}, col sep=comma] {fig11.csv};
	\addplot[color=blue, mark=triangle, solid] table [x={ndof13}, y={error13}, col sep=comma] {fig11.csv};
	\addplot[color=red, mark=o, mark options={solid},dashed] table [x={ndof21}, y={error21}, col sep=comma] {fig11.csv};
	\addplot[color=red, mark=square, mark options={solid},dashed] table [x={ndof22}, y={error22}, col sep=comma] {fig11.csv};
	\addplot[color=red, mark=triangle, mark options={solid},dashed] table [x={ndof23}, y={error23}, col sep=comma] {fig11.csv};
	\addplot[color=black, mark=*, solid] table [x={ndof31}, y={error31}, col sep=comma] {fig11.csv};
	\addplot[color=black, mark=square*, solid] table [x={ndof32}, y={error32}, col sep=comma] {fig11.csv};
	\addplot[color=black, mark=triangle*, solid] table [x={ndof33}, y={error33}, col sep=comma] {fig11.csv};
	\legend{{$p_{\text{FEM}}=1,\,q=1$},{$p_{\text{FEM}}=1,\,q=2$},{$p_{\text{FEM}}=1,\,q=3$},
		{$p_{\text{FEM}}=2,\,q=1$},{$p_{\text{FEM}}=2,\,q=2$},{$p_{\text{FEM}}=2,\,q=3$},
		{$p_{\text{FEM}}=3,\,q=1$},{$p_{\text{FEM}}=3,\,q=2$},{$p_{\text{FEM}}=3,\,q=3$}};
	\end{loglogaxis}
	\end{tikzpicture}
	\caption{Mortaring of two non-conforming multipatch unit cube domains: 10th eigenvalue convergence for different choices of the discretisation degrees.} \label{fig:non-conforming-convergence}
\end{figure}
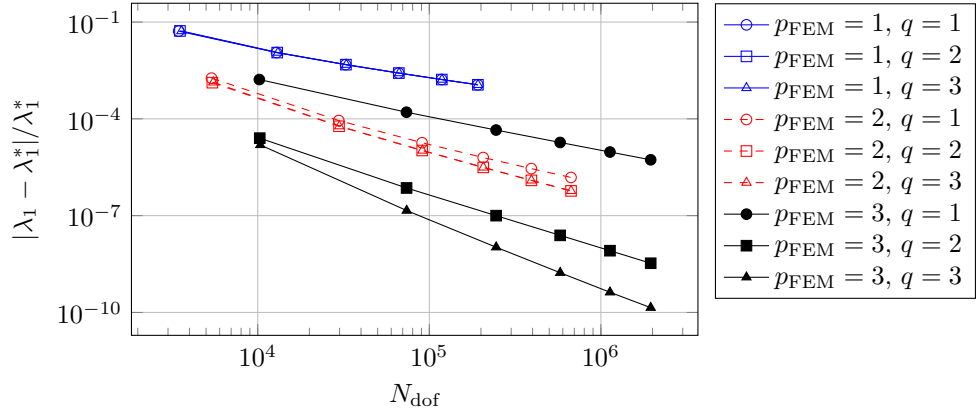

\subsection{State Space Concatenation}
For the \gls{ssc} method we perform analogous tests as for the mortar case. Given we do not have a proof for the stability of the coupling, we are particularly interested in investigating numerically the behaviour of the solution with respect to the number of waveguide modes selected as Lagrange multipliers. 

\paragraph{Cube with two patches} Let us consider the two-patch geometry in Fig.~\ref{fig:2cubes_geo}. We discretise with \gls{iga} on each subdomain, and since the interface is a square, the waveguide eigenmodes $\vect{\varphi}_k$ can be computed analytically {\cite{Flisgen_2015aa}}. Figure~\ref{fig:SSC_2cubes_conv} shows the convergence of the first eigenvalue to the exact solution for the case of matching and non-matching grids, and for different choices of the discretisation degrees, while keeping fixed the number of waveguide modes $\Ngamma=18$.

\begin{figure}
	\centering
		\begin{tikzpicture}
		\begin{loglogaxis}[width=.7\textwidth,height=6cm,xlabel={$h$},ylabel={$|\lambda-\lambda^*| / \lambda^*$},xmin=1e-2,xmax=1,ymin=5e-9,ymax=5e-3,legend pos=outer north east, legend style={cells={align=left}},grid = major]
		\addplot[red, mark=*] table [x=h, y=mode1, col sep=comma] {fig12a.csv};
		\addlegendentry{same degree $p$,\\matching}
		\addplot[red,dashed,mark=o, mark options={solid}] table [x=h, y=mode1, col sep=comma] {fig12b.csv};
		\addlegendentry{same degree $p$,\\non-matching}
		\addplot[blue,mark=square*] table [x=h, y=mode1, col sep=comma] {fig12c.csv};
		\addlegendentry{different degrees,\\matching}
		\addplot[blue, dashed, mark=square, mark options={solid}] table [x=h, y=mode1, col sep=comma] {fig12d.csv};
		\addlegendentry{different degrees,\\non-matching}
		\end{loglogaxis}
		\end{tikzpicture}
	\caption{SSC coupling of two conforming unit cube patches: first eigenvalue convergence for different choices of the discretisation degrees and maximum regularity, i.e., $p=p_1=p_2=2$, $r=r_1=r_2=1$ (red circles) and $p_1=3$, $p_2=2$, $r_1=2$, $r_2=1$ (blue squares), and for matching and non-matching grids on $\Gamma$ (dashed and non-dashed lines). The number of waveguide modes is fixed to $\Ngamma=18$.\label{fig:SSC_2cubes_conv}}
\end{figure}
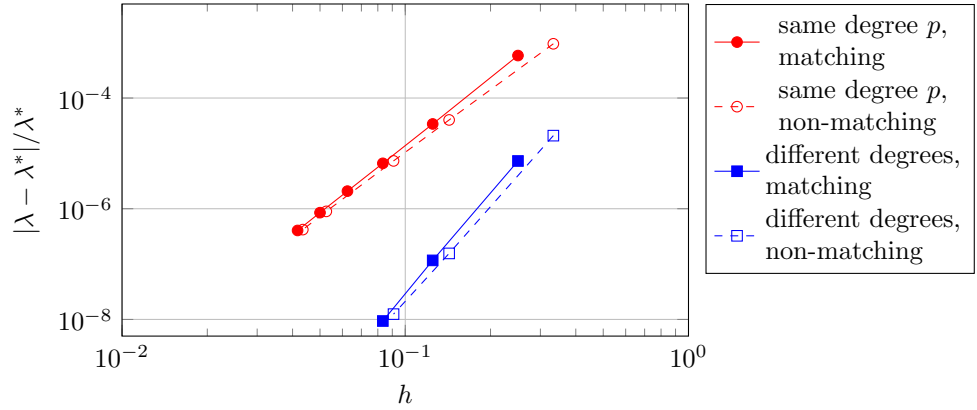

In Fig.~\ref{fig:SSC_spectrum} we present the relative errors of the first \num{20} computed eigenfrequencies in the cube obtained with a fixed B-Spline discretisation on both sides ($p_1 = 3$, $r_1 = 2$ and $p_2 = 2$, $r_2 = 1$ with non-matching grids on the interface) while increasing the number of analytical waveguide modes on the interface. It is noticeable how $\Ngamma$ influences the spectrum approximation, in particular, when not enough modes are chosen, since some eigenfunctions cannot be represented by the Lagrange multiplier, some of the higher order modes are not correctly captured. However the size of the coupling space cannot be taken arbitrarily big since the saddle point becomes unstable. This is illustrated in Fig.~\ref{fig:SSC_beta}, where the $\betainfsup$ constant is approximated for different choices of $\Ngamma$ using the numerical test from \cite{Chapelle_1993aa}. It is evident that increasing $\Ngamma$ causes the method to fail if the two subdomains are not refined accordingly.

\begin{figure}[!t]
	\begin{minipage}[t]{.48\textwidth}
		\begin{tikzpicture}
		\begin{semilogyaxis}[width=\textwidth,xlabel={Mode},ylabel={$|\lambda-\lambda^*| / \lambda^*$},ymin=1e-6,ymax=1,grid=major]
		\addplot[mark=*,only marks] table [x=Mode, y=err, col sep=comma] {fig13a.csv};
		\end{semilogyaxis}
		\end{tikzpicture}
		\subcaption{Number of waveguide modes $\Ngamma = 2$.}
	\end{minipage}
	\hfill
	\begin{minipage}[t]{.48\textwidth}
		\begin{tikzpicture}
		\begin{semilogyaxis}[width=\textwidth,xlabel={Mode},ylabel={$|\lambda-\lambda^*| / \lambda^*$},ymin=1e-6,ymax=1,grid=major]
		\addplot[mark=*,only marks] table [x=Mode, y=err, col sep=comma] {fig13b.csv};
		\end{semilogyaxis}
		\end{tikzpicture}
		\subcaption{Number of waveguide modes $\Ngamma = 6$.}
	\end{minipage}
	
	\begin{minipage}[t]{.48\textwidth}
		\begin{tikzpicture}
		\begin{semilogyaxis}[width=\textwidth,xlabel={Mode},ylabel={$|\lambda-\lambda^*| / \lambda^*$},ymin=1e-6,ymax=1,grid=major]
		\addplot[mark=*,only marks] table [x=Mode, y=err, col sep=comma] {fig13c.csv};
		\end{semilogyaxis}
		\end{tikzpicture}
		\subcaption{Number of waveguide modes $\Ngamma = 18$.}
	\end{minipage}
	\hfill
	\begin{minipage}[t]{.48\textwidth}
		\begin{tikzpicture}
		\begin{semilogyaxis}[width=\textwidth,xlabel={Mode},ylabel={$|\lambda-\lambda^*| / \lambda^*$},ymin=1e-6,ymax=1,grid=major]
		\addplot[mark=*,only marks] table [x=Mode, y=err, col sep=comma] {fig13d.csv};
		\end{semilogyaxis}
		\end{tikzpicture}
		\subcaption{Number of waveguide modes $\Ngamma = 34$.}
	\end{minipage}
	\caption{SSC coupling between two patches: convergence of the first \num{20} eigenvalues for a fixed discretisation on the two subdomains and an increasing number of waveguide modes $\Ngamma$.\label{fig:SSC_spectrum}}
\end{figure}
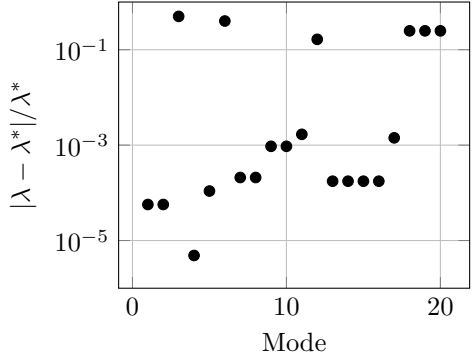
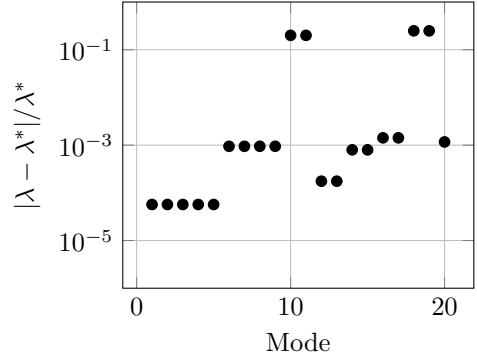
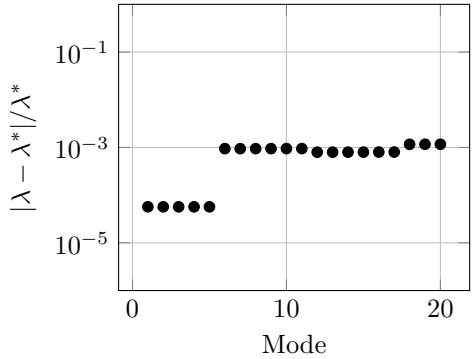
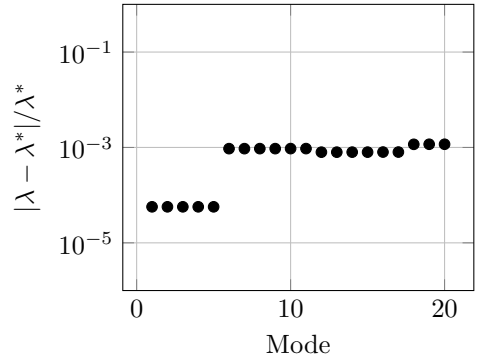

\begin{figure}[!t]
	\begin{minipage}[t]{.48\textwidth}
		\begin{tikzpicture}
		\begin{loglogaxis}[width=\linewidth,xlabel={$h$},ylabel={$\betainfsup$},xmin=1e-2,xmax=3e-1,ymin=1e-8,ymax=.5,legend pos=south west,grid=major]
		\addplot table [x=h1, y=beta1, col sep=comma] {fig14a.csv};
		\addplot table [x=h15, y=beta15, col sep=comma] {fig14b.csv};
		\addplot table [x=h26, y=beta26, col sep=comma] {fig14c.csv};
		\legend{$\Ngamma=1$,$\Ngamma=15$,$\Ngamma=26$}
		\end{loglogaxis}
		\end{tikzpicture}
		\caption{SSC coupling of two patches: $\betainfsup$ constants for an increasing number of modes, instability for $\Ngamma\to\infty$. 		\label{fig:SSC_beta}}
	\end{minipage}
	\hfill
	\begin{minipage}[t]{.48\textwidth}
		\begin{tikzpicture}
		\begin{semilogyaxis}[width=\linewidth,xlabel={Mode},ylabel={$|\lambda-\lambda^*| / \lambda^*$},grid=major]
		\addplot[mark=*,only marks] table [x=mode, y=err, col sep=comma] {fig15.csv};
		\end{semilogyaxis}
		\end{tikzpicture}
		\caption{SSC coupling of \gls{iga} ($p_1 = 3$, $r_1=2$, $N_\Gamma=25$) and tetrahedral lowest-order \gls{fem}: spectrum approximation.\label{fig:SSC_IGA-FEM}}
	\end{minipage}
\end{figure}

As mentioned in section~\ref{sec:mortar_vs_ssc}, the \gls{ssc} coupling allows for straightforward coupling of completely different grids, since the construction of the coupling matrices is completely independent on each side. In Fig.~\ref{fig:SSC_IGA-FEM} the approximation of the first 40 eigenvalues in the cube for an \gls{iga}-\gls{fem} coupling is shown. Domain $ \Omega_1 $ is discretised with \gls{iga} ($ p_1=3 $, $ r_1=2 $, $ \Ndof \approx \num{500} $), while domain $ \Omega_2 $ employs classical first order tetrahedral edge elements \gls{fem} ($ \Ndof \approx \num{50000} $). The main advantage here is that no computation of the intersection mesh is required.

\paragraph{Pillbox geometry} We then consider the pillbox geometry showed in Fig.~\ref{fig:2pillbox_geo}. The interface $\Gamma$ is a circle, thus the closed form solutions for the waveguide modes (see \cite{Hill_2009aa}) can be used to exactly evaluate the waveguide modes $\vect{\varphi}_k$. We use both the \gls{te} and \gls{tm} modes as the basis. One side is discretised with \gls{iga} using basis functions of degree two and regularity $r=1$ ($\Ndof= 5440$), while the other side is discretised with \gls{fem} using low order N\'ed\'elec edge elements ($\Ndof=6322$), and we set $N_\Gamma = 25$. The results for the computed eigenfrequencies are reported in Table~\ref{tab:SSC_pillbox_result} along with the exact values. It is evident that some spurious modes appear in the spectrum as a consequence of the coupling due to non-physical charge appearing on the interface $\Gamma$, see Fig.~\ref{fig:pillbox_unphysical}, where we plot the magnitude of one of the modes associated to a spurious eigenvalue. 

\begin{figure}
\centering
\includegraphics[width=.5\textwidth]{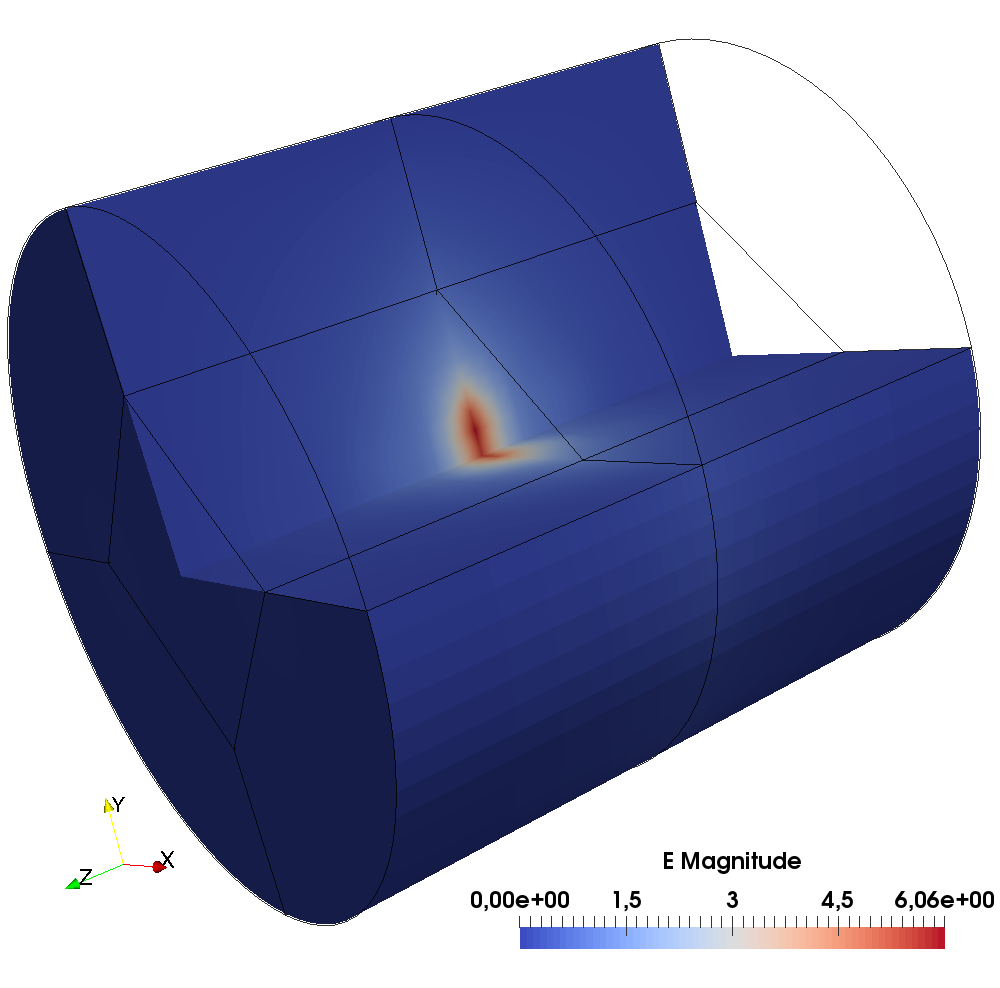}
\caption{An example of one of the modes presenting unphysical charge on the coupling interface computed using \gls{ssc}.\label{fig:pillbox_unphysical}} 
\end{figure}\hfill

\begin{table}[!t]
	\centering
	\setlength{\tabcolsep}{0.5em}
	\begin{minipage}{.45\textwidth}
		\centering
		\begin{tabular}{lll}
			\toprule
			$f_\text{exact}$ & $f_\text{mortar}$ & $f_\text{SSC}$\\
			\midrule
			                &                & \num{0.035857} \\
			                &                & \num{0.035857} \\
			                &                & \num{0.577485} \\
			                &                & \num{0.584617} \\
			                &                & \num{1.841232} \\
			                &                & \num{1.841232} \\
			                &                & \num{2.253919} \\
			                &                & \num{2.254197} \\
			 \num{2.373958} & \num{2.373960} & \num{2.373968} \\
			 \num{2.373958} & \num{2.373960} & \num{2.373968} \\
			 \num{2.705705} & \num{2.705706} & \num{2.704475} \\
			 \num{2.705705} & \num{2.705706} & \num{2.704475} \\
			\bottomrule
		\end{tabular}
	\end{minipage}	\hspace{.05\textwidth}	\begin{minipage}{.45\textwidth}
		\centering
		\begin{tabular}{lll}
			\toprule
			$f_\text{exact}$ & $f_\text{mortar}$ & $f_\text{SSC}$\\
			\midrule
			 \num{2.942116} & \num{2.942116} & \num{2.942116}\\
			 \num{3.036078} & \num{3.036078} & \num{3.036078}\\
			 \num{3.182680} & \num{3.182681} & \num{3.182785}\\
			 \num{3.182680} & \num{3.182681} & \num{3.182785}\\
			                &                & \num{3.214081}\\
			                &                & \num{3.214081}\\
			 \num{3.301959} & \num{3.301961} & \num{3.301962}\\
			 \num{3.702910} & \num{3.702919} & \num{3.702994}\\
			 \num{3.749870} & \num{3.749879} & \num{3.753896}\\
			 \num{3.749870} & \num{3.749879} & \num{3.753896}\\
			 \num{3.811044} & \num{3.811087} & \num{3.811201}\\
			 \num{3.811044} & \num{3.811087} & \num{3.811263}\\
			\bottomrule
		\end{tabular}
	\end{minipage}
	\caption{Comparison between the exact eigenfrequencies in \si{\giga\hertz} of the pillbox cavity and the ones computed using mortaring and \gls{ssc}. For mortar we present the case $p_{\text{IGA}}=4$, $q=3$, $p_{\text{FEM}}=1$ ($\Ndof \approx 16000$), while for SSC we chose $p_1 = 2$, $p_2=1$, $N_\Gamma=25$ ($\Ndof \approx 7000$).\label{tab:SSC_pillbox_result}}
\end{table}

\subsubsection{Simulation of a full TESLA cavity}

\begin{figure}
\begin{center}
\includegraphics[width=.9\linewidth]{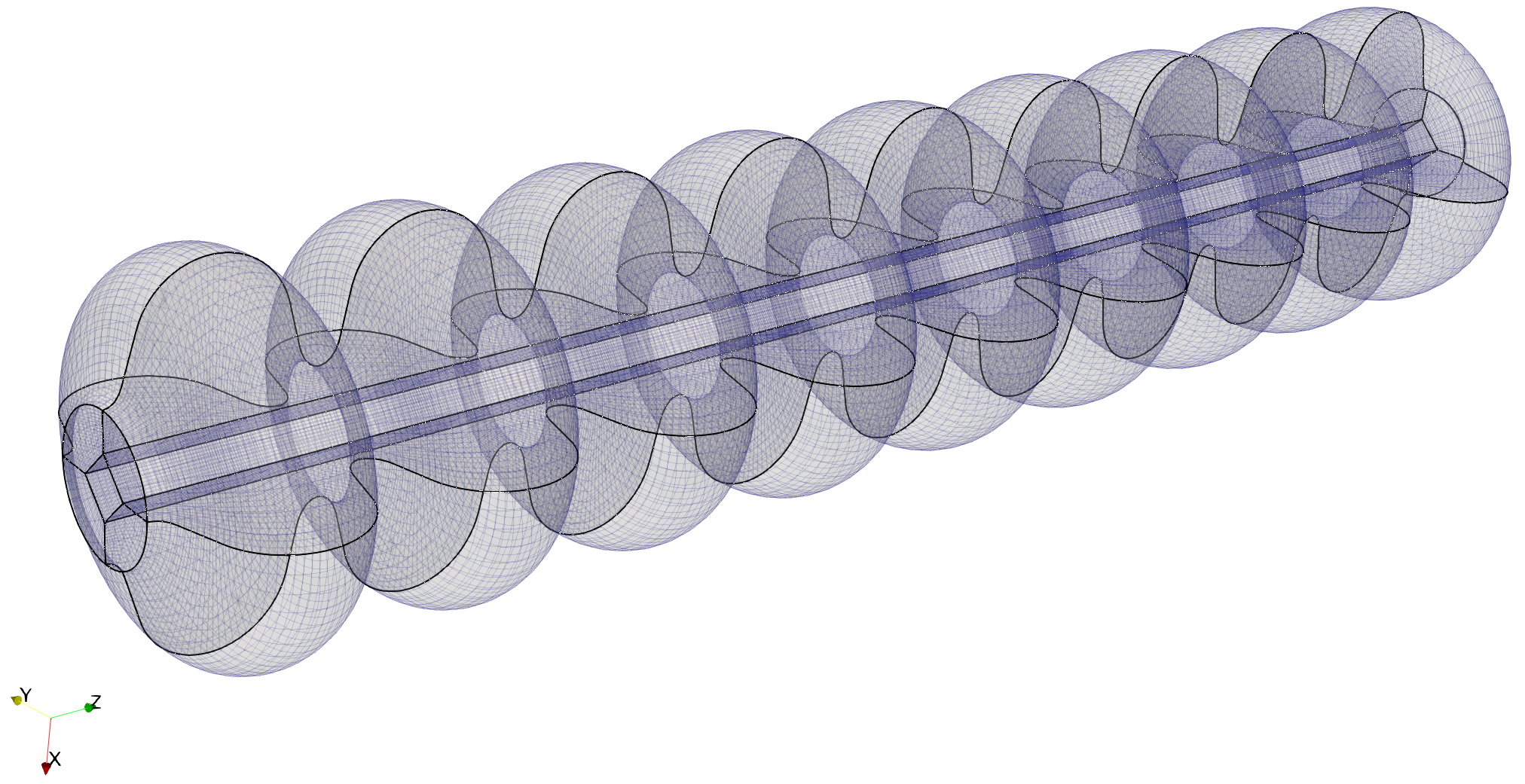}
\caption{Patch subdivision (black lines) and mesh (blue lines) for the IGA section of the full TESLA cavity model, the mesh has 440020 subdivisions.}
\label{fig:fullteslamesh}
\end{center}
\end{figure}

As a final example of the applicability of the two coupling methods to \gls{rf} cavity simulation, we consider the \gls{tesla} cavity, including the two \gls{hom} couplers at both ends (see Fig.~\ref{fig:homc}). We consider the cavity as if composed by \num{11} blocks (\num{7} of which are identical mid cells whose matrices can be assembled only once) separated by \num{10} circular interfaces \cite{TTF_1995}. Each cell is discretised with \gls{iga} using second degree basis functions and approximately \num{40000} degrees of freedom per cell, the mesh of the cells is shown in Fig.~\ref{fig:fullteslamesh}. The coupling between the cells is performed using Mortar with $q=1$. The two beampipes with the \gls{homc} are instead triangulated by tetrahedra and the discrete matrices are assembled using lowest order N\'ed\'elec Finite Elements through an in-house code (approximately \num{75000} elements). The coupling of the cavity with the beampipes is performed using the \gls{ssc} technique since, as showed before, it is easier to construct the coupling matrices without the necessity of an intersection mesh. In Fig.~\ref{fig:IGA-FEM_approach} the enforced subdivision is highlighted.

As a proof of concept we apply Dirichlet, or equivalently \gls{pec} boundary conditions at the couplers. An even more realistic simulation would impose port boundary conditions. The results are reported in Table~\ref{tab:tesla9_eigenmodes_full} where it is possible to see the presence of nine spurious modes at the beginning of the spectrum when comparing with a finite element reference computation.

\begin{figure}[t]
	\includegraphics[width=\textwidth]{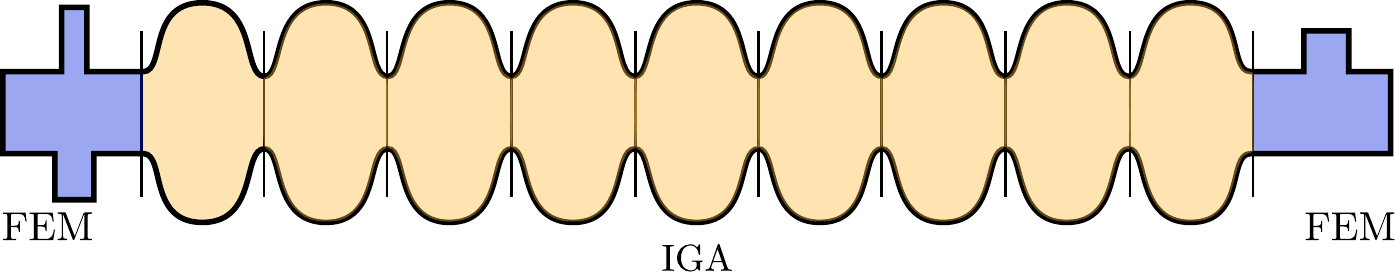}
	\caption{Substructuring with IGA-FEM for full cavity simulation. Each cell is discretized independently with \gls{iga} (orange colour), the two end couplers with \gls{fem} (in blue). All the \gls{iga} pieces are coupled together through mortar method, while the \gls{iga} and \gls{fem} pieces are coupled together through the \gls{ssc} method.}\label{fig:IGA-FEM_approach}
\end{figure}

\begin{table}[!t]
	\centering
	\begin{minipage}{.4\textwidth}
		\centering
		\begin{tabular}{ll}
			\toprule
			$f_\text{FEM}$ & $f_\text{DDM}$\\
			\midrule
			               & \num{0.049932}\\
			               & \num{0.269500}\\
			               & \num{0.274693}\\
			               & \num{0.296580}\\
			               & \num{0.301146}\\
			               & \num{0.464058}\\
			               & \num{0.494598}\\
			               & \num{0.509495}\\
			               & \num{0.529445}\\
			\num{1.276281} & \num{1.277173}\\
			\num{1.278328} & \num{1.279250}\\
			\num{1.281485} & \num{1.282478}\\
			\num{1.285378} & \num{1.286496}\\
			\num{1.289557} & \num{1.290818}\\
			\num{1.293506} & \num{1.294900}\\
			\num{1.296745} & \num{1.298226}\\
			\num{1.298891} & \num{1.300387}\\
			\num{1.299585} & \num{1.301132}\\
			\num{1.622902} & \num{1.622122}\\
			\num{1.623323} & \num{1.622139}\\
			\bottomrule
		\end{tabular}
	\end{minipage}
	~
	\begin{minipage}{.4\textwidth}
		\centering
		\begin{tabular}{ll}
			\toprule
			$f_\text{FEM}$ & $f_\text{DDM}$\\
			\midrule
			\num{1.630347} & \num{1.629433}\\
			\num{1.630735} & \num{1.629452}\\
			\num{1.642956} & \num{1.641527}\\
			\num{1.643286} & \num{1.641594}\\
			\num{1.660529} & \num{1.658129}\\
			\num{1.660681} & \num{1.658379}\\
			\num{1.682429} & \num{1.678763}\\
			\num{1.682568} & \num{1.679235}\\
			\num{1.707649} & \num{1.702588}\\
			\num{1.707769} & \num{1.703258}\\
			\num{1.734849} & \num{1.728442}\\
			\num{1.734901} & \num{1.729246}\\
			\num{1.762429} & \num{1.754965}\\
			\num{1.762557} & \num{1.755841}\\
			\num{1.789472} & \num{1.781015}\\
			\num{1.789628} & \num{1.782023}\\
			\num{1.800177} & \num{1.800734}\\
			\num{1.800241} & \num{1.800779}\\
			\num{1.837903} & \num{1.823147}\\
			\num{1.838053} & \num{1.824116}\\
			\bottomrule	
		\end{tabular}
	\end{minipage}
	\caption{First 40 computed eigenfrequencies of the \gls{tesla} cavity including the \gls{homc} in \si{\giga\hertz}. The simulation of $f_\text{DDM}$ is performed with the same substructuring as in Fig.~\ref{fig:IGA-FEM_approach} while $f_\text{FEM}$ is determined by a Finite Element reference computation using CST Microwave Studio (second order basis functions, default settings), \cite{CST_2018aa}.
Some spurious modes appear at the beginning of the spectrum.}\label{tab:tesla9_eigenmodes_full}
\end{table}
 
\section{Conclusions}
\label{sec:conclusions}

This paper discussed two substructuring approaches that allow the convenient coupling of subdomains discretized by IGA with any other method, in particular FEM. It was shown by numerical examples that a modal basis for the Lagrange multiplier space allows for an easy implementation but is not stable and may cause spurious modes. On the other hand, the isogeometric mortaring is proven to be spectral correct if the degree $q$ on the interface is properly chosen, i.e., $q=p-1$ where $p$ is the spline degree of the slave domain. Numerical examples underline those findings. 

\section*{Acknowledgments}
This work was supported by the Excellence Initiative of the German Federal and State Governments and the Graduate School of Computational Engineering at Technische Universit\"at Darmstadt and the DFG grant SCHO1562/3-1. The work of A.B and R.V has been partially supported by the ERC Advanced Grant ``CHANGE'' (694515,  2016-2020).

\end{document}